\documentclass[12pt]{article}

\usepackage{amssymb, amsmath, amsthm}
\usepackage{mathrsfs}
\usepackage[left=2cm,top=4cm,bottom=2.5cm,right=2cm]{geometry}
\usepackage[utf8]{inputenc} 
\usepackage[T1]{fontenc}    
\usepackage{hyperref}       
\usepackage{url}            
\usepackage{booktabs}       
\usepackage{amsfonts}       
\usepackage{nicefrac}       
\usepackage{microtype}      
\usepackage{xcolor}         
\usepackage{physics}        
\usepackage{amssymb, amsmath, amsthm}
\usepackage{xpatch}
\usepackage[pdftex]{graphicx}
\usepackage{authblk}
\usepackage[symbol]{footmisc}
\usepackage{multirow}       
\usepackage{diagbox}       
\usepackage{makecell}
\usepackage{caption}
\usepackage{enumitem}
\usepackage{comment}
\usepackage{natbib}         

\makeatletter
\AtBeginDocument{\xpatchcmd{\@thm}{\thm@headpunct{.}}{\thm@headpunct{}}{}{}}
\makeatother

\newcommand{\mathbbm}[1]{\text{\usefont{U}{bbm}{m}{n}#1}}

\DeclareMathOperator\arctanh{arctanh}
\DeclareMathOperator{\Var}{Var}
\DeclareMathOperator{\diag}{diag}
\DeclareMathOperator{\linspan}{span}

\newtheorem{theorem}{Theorem}
\newtheorem*{theorem*}{Theorem}
\newtheorem{proposition}[theorem]{Proposition}
\newtheorem{corollary}[theorem]{Corollary}
\newtheorem{lemma}{Lemma}
\newtheorem*{lemma*}{Lemma}
\newtheorem{definition}{Definition}
\newtheorem*{definition*}{Definition}
\newtheorem{remark}{Remark}
\newtheorem*{remark*}{Remark}

\newtheorem*{fact*}{Fact}

\newtheoremstyle{theoremnoperiod}
    {\medskipamount}%
    {\medskipamount}%
    {\itshape}%
    {\parindent}%
    {\scshape}%
    {}%
    {1em}%
    {}%
\theoremstyle{theoremnoperiod}
\newtheorem*{Restatementlemma*}{Restatement of Lemma\hspace*{-9pt}}  
\newtheorem*{Restatementtheorem*}{Restatement of Theorem\hspace*{-9pt}}
\newtheorem*{Restatementproposition*}{Restatement of Proposition\hspace*{-9pt}}

\begin{document}

\title{Optimal gentle measurements of finite-dimensional quantum states}
\date{}

\author[1]{\small Cristina Butucea \footnote{\texttt{cristina.butucea@ensae.fr}}}
\author[2]{\small Jan Johannes \footnote{\texttt{johannes@math.uni-heidelberg.de}}}
\author[1, 2]{\small Henning Stein \footnote{\texttt{henning.stein@math.uni-heidelberg.de}}}

\affil[1]{\footnotesize CREST, ENSAE, Institut Polytechnique de Paris, 91120 Palaiseau, France}
\affil[2]{\footnotesize Heidelberg University, 69120 Heidelberg, Germany}

\maketitle

\begin{abstract}
    Standard approaches to quantum statistical inference rely on measurements that induce a collapse of the wave function, effectively consuming the quantum state to extract information. In this work, we investigate the fundamental limits of \emph{locally-gentle} quantum state certification, where the learning algorithm is constrained to perturb the state by at most $\alpha$ in trace norm, thereby allowing for the reuse of samples. We analyze the hypothesis testing problem of distinguishing whether an unknown state $\rho$ is equal to a reference state $\rho_0$ or $\epsilon$-far from it. We derive the minimax sample complexity for this problem, quantifying the information-theoretic price of non-destructive measurements. Specifically, by constructing explicit measurement operators, we show that the constraint of $\alpha$-gentleness imposes a sample size penalty of $\frac{d}{\alpha^2}$, yielding a total sample complexity of $n = \Theta(\frac{d^3}{\epsilon^2 \alpha^2})$. Our results clarify the trade-off between information extraction and state disturbance, and highlight deep connections between physical measurement constraints and privacy mechanisms in quantum learning. Crucially, we find that the sample size penalty incurred by enforcing $\alpha$-gentleness scales linearly with the Hilbert-space dimension $d$ rather than the number of parameters $d^2-1$ typical for high-dimensional private estimation. 
\end{abstract}

\section{Introduction}
A central postulate of quantum mechanics dictates that observation inevitably alters the observed system. In the standard regime of projective measurements, this disturbance is maximal, resulting in the "collapse of the wave function". Consequently, the quantum state is effectively consumed upon measurement, preventing any further extraction of information from that specific copy. While this destructive nature is assumed in most quantum testing algorithms, it is not an absolute necessity.

In this work, we depart from the destructive paradigm and investigate the challenge of locally-gentle quantum state certification. Here, the learning algorithm is constrained to perform measurements that are minimally invasive. Formally, we aim to learn an unknown quantum state $\rho \in \mathbb{C}^{d \times d}$ given access to $n$ copies. We consider the hypothesis testing task:
\begin{equation}
\label{eqn::state_certification_problem}
    H_0: \rho = \rho_0 \hspace{20pt} \text{vs} \hspace{20pt} H_1: \norm{\rho - \rho_0}_{Tr} > \epsilon,
\end{equation}
where $\norm{\cdot}_{Tr}$ denotes the trace norm. Crucially, we enforce that the measurements are $\alpha$-gentle. That is, for a gentleness parameter $\alpha$, the measurement ensures that the distance between the pre-measurement state $\rho$ and the post-measurement state $\rho_{M \to y}$ is bounded by $\norm{\rho - \rho_{M \to y}}_{Tr} \leq \alpha$ for all possible outcomes $y$.

Unlike the standard setting where the state collapses to an eigenstate, this constraint ensures the state is only altered $\alpha$-minimally and can therefore be reused for subsequent information extraction. It has been suggested (\citet{Abbas2023}) that this preservation of quantum information is essential for the efficient implementation of quantum backpropagation algorithms. The central objective of this paper is to determine the information-theoretic price of this preservation.
\subsection{Related Work.} 
In recent years, in part thanks to \citet{artiles_invitation_2005}, quantum statistics has received a surge of interest from statisticians who aim to understand the inherently random nature of quantum systems. While quantum systems do have fundamentally different properties to classical systems \citet{NielsenChuang}, statisticians were fast to discover that some tools of classical statistics allow for a very fruitful application in the quantum world. In particular, it has been shown that the measurement postulate of quantum mechanics, which dictates that any quantum system can only be observed indirectly, is closely related to the concept of compressed sensing \citet{wang_asymptotic_2013}. Since quantum states can be described by positive semi-definite hermitian matrices, many quantum algorithms arose from the better understood field of matrix recovery using compressed sensing such as \citet{gross_quantum_2010, koltchinskii_von_2011, flammia_quantum_2012, kueng_low_2017,  carpentier_uncertainty_2019} and their adaptation to the quantum setting have proven optimal not only for the estimation problem \citet{koltchinskii_optimal_2015, haah_sample-optimal_2017, guta2018faststatetomographyoptimal} but also for quantum testing \citet{yu:LIPIcs.ITCS.2021.11, liu2024rolesharedrandomnessquantum}. In order for us to study the effect of locally gentle measurements for quantum state certification, we combine these statistical insights with the study of gentle measurements and, interestingly, differential privacy. \\

\textbf{State Certification.}
Quantum state certification is a special task of quantum property testing (\citet{montanaro_survey_2016}) in which one aims to learn properties of a quantum state by means of quantum measurements. This field of research has seen a great amount of interest in recent years with the works of \citet{odonnell_quantum_2021, badescu_quantum_2019, Bubeck2020, yu:LIPIcs.ITCS.2021.11} and \citet{liu2024rolesharedrandomnessquantum}. While \citet{odonnell_quantum_2021} and \citet{badescu_quantum_2019} consider the theoretically more efficient regime of entangled measurements, given their current experimental unfeasibility (\citet{cotler_quantum_2020}), the focus has shifted in recent years to develop sample optimal algorithms for unentangled (product) measurements. \citet{Bubeck2020} showed that for unentangled randomized measurements, a total of $n = \Theta(\frac{d^{3/2}}{\epsilon^2})$ copies of $\rho$ are needed to distinguish it from the maximally mixed state. Later, \citet{yu:LIPIcs.ITCS.2021.11} provided an algorithm using mutually unbiased bases that achieves state certification using $n = O(\frac{d^2}{\epsilon^2})$ copies with fixed measurements. \citet{liu2024rolesharedrandomnessquantum} subsequently proved the corresponding lower bound $n = \Omega(\frac{d^2}{\epsilon^2})$ for fixed unentangled measurements. Furthermore, their proof technique recovered the lower bound from \citet{Bubeck2020} as a corollary. They also provided a generalized upper bound algorithm utilizing 2-designs, which coincides with that of \citet{yu:LIPIcs.ITCS.2021.11} when $d$ is a prime power. \\

\quad \textbf{Gentleness.} Prior to recent developments, the literature on the statistical properties of gentle measurements was sparse, extending little beyond the rudimentary bounds provided by the gentle measurement lemma \citet{Winter_1999}. 
A significant shift occurred when \citet{aaronson2019gentle} established a connection between gentleness and differential privacy, proposing a gentle algorithm for shadow tomography—a task distinct from ours where the goal is to learn only some derived properties of the state rather than the state itself. Recently, \citet{butucea2025sampleoptimallearningquantumstates} analyzed locally-gentle state tomography and certification specifically for qubits. They established that for these single-qubit tasks, gentleness incurs a multiplicative penalty of $\frac{1}{\alpha^2}$, resulting in a sample optimal rate of $n = \Theta(\frac{1}{\epsilon^2 \alpha^2})$. \\

\textbf{Connections to Differential Privacy.} It was observed by \citet{aaronson2019gentle} that a global gentle measurement can be constructed such that the resulting statistics follow the Laplace mechanism known from classical differential privacy \citet{Dwork2006}. While \citet{aaronson2019gentle} discussed approximate implementations of the aforementioned mechanism, \citet{butucea2025sampleoptimallearningquantumstates} constructed an implementable physical measurement, gentle on each system component, where the outcomes follow the label switching privacy mechanism (\citet{steinberger2023efficiency}). Given the prominence of differential privacy in statistical inference \citet{amorino_minimax_2025, berrett_locally_2020, duchi2014localprivacydataprocessing, kent_locally_2026}, a major challenge lies in upgrading classical privacy mechanisms into gentle measurements that preserve the quantum state while providing optimal guarantees for the inference problem at hand.

\section{Main Results}

We now state our main result, which establishes the fundamental limit for quantum state certification under local gentleness constraints and give a short description of the main technical difficulties of our contributions. The following Theorem is a consequence of the more general upper bound in Theorem \ref{thm::upper_bound} and the matching lower bound in Theorem \ref{thm::lower_bound}, for locally-$\alpha$-gentle quantum state certification in a simplified form.

\begin{theorem*}[Minimax Sample Complexity]
    Consider the testing task $H_0: \rho = \rho_0$ versus $H_1: \norm{\rho - \rho_0}_{\mathrm{tr}} > \epsilon$, where $\rho_0 = \frac{1}{d}\mathbbm{1}$ is the maximally mixed state. Then, a total number of 
    \begin{equation*}
        n = \Theta\left( \frac{d^3}{\alpha^2 \epsilon^2} \right)
    \end{equation*}
    copies are needed and sufficient to achieve a success probability of at least $2/3$ when restricting ourselves to fixed, unentangled, locally $\alpha$-gentle measurements,.
\end{theorem*}

\textbf{The Price of Gentleness.}
Comparing our result to the standard unentangled setting, where $n = \Theta(d^2/\epsilon^2)$ is sufficient for fixed measurements (\citet{yu:LIPIcs.ITCS.2021.11}), we observe that enforcing $\alpha$-gentleness incurs a multiplicative penalty of $d/\alpha^2$.
This scaling highlights a surprising efficiency in quantum gentle learning. The parameter space of a $d$-dimensional quantum state has dimension $d^2-1$. In classical differential privacy, the sample complexity penalty typically scales linearly with the dimension of the parameter space (i.e., one might expect a penalty of order $d^2$) 
However, we find that the quantum ``price of gentleness'' scales only as $ d$.
This separation suggests that ``lifting'' classical privacy mechanisms to the quantum regime, that is converting stochastic maps into physical measurement operators, yields fundamentally different statistical behaviors due to the specific geometric properties of quantum state space. \\

\textbf{Constructive Upper Bound via Noisy measurements.}
A major challenge in gentle learning is ``lifting'' classical private mechanisms (which act on probability distributions) into valid quantum instruments (which act on density matrices). We propose an explicit construction of measurement operators that satisfy the gentleness constraint while preserving statistical utility.
Our construction utilizes mutually unbiased bases, which are sets of orthonormal bases that are equally spaced apart in the $d$ dimensional Hilbert space. Aiming at generalizing the classical RAPPOR mechanism for high-dimensional private testing (\citet{acharya2021inferenceinformationconstraintsiii}), we define a ``noisy'' measurement operator that is both gentle and statistically optimal. For a set of mutually unbiased bases comprised of vectors $(|e_k^{(b)}\rangle)_{k = 1,...,d; b = 1,...,d+1}$, a vector $z \in \{0,1\}^z$, and a noise parameter $\delta > 0$, we define the operator:
\begin{equation*}
    E_{\delta, z}^{(b)} = \left( \frac{e^{\delta/2}}{e^{\delta/2} + 1} \right)^d  \sum_{k = 1}^d e^{-\frac{\delta}{2} \norm{z - e_k}_1} \ket{e_k^{(b)}}\bra{e_k^{(b)}},
\end{equation*}
where $e_k \in \{0,1\}^d$ is the standard basis vector for the $k$-th entry. We prove that for appropriate $\delta$, the collection $\{E_{\delta, z}\}_z^{(b)}$ forms a valid POVM that is $\alpha$-gentle. Furthermore, we show that classical post-processing of the outcomes of this POVM achieves the upper bound. \\

\textbf{Lower Bound Framework for Full-Rank Measurements.}
To establish optimality, we extend the lower bound framework of \citet{liu2024rolesharedrandomnessquantum}. A critical technical obstacle arises from the nature of gentleness: as noted in \citet{butucea2025sampleoptimallearningquantumstates}, gentle measurements must be full-rank. Otherwise, input states near the operator's null space would suffer total collapse. Consequently, we cannot rely on standard techniques that assume rank-one POVMs. 

We overcome this by analyzing the specific structure of the information loss. For a gentle measurement defined by POVM elements $(E_y)_{y \in \mathcal{Y}}$, we characterize the $\chi^2$-fluctuation around the distribution induced by $\rho_0$ via the linear super-operator $\mathcal{H}: \mathbb{C}^{d \times d} \to \mathbb{C}^{d \times d}$, defined as:
\begin{equation}
\label{eqn::defn_linear_superoperator}
    \mathcal{H}(A) := \sum_{y \in \mathcal{Y}} \frac{\Tr[ A E_{y} ]}{\Tr[E_{y}]} E_{y}.
\end{equation}
We show that $\mathcal{H}$ captures the relevant statistical properties of general (non-rank-one) measurements, serving as a generalization of the analysis of the Lüders channel. By proving that $\mathcal{H}$ is self-adjoint, we can identify the least sensitive directions in the state space along which to perturb $\rho_0$ corresponding to the smallest eigenvalues of $\mathcal{H}$. Constructing local perturbations along these directions yields the matching lower bound.

\begin{table}[ht!]
\centering
\begin{tabular}{l|c|c}
    
     & \textbf{non-gentle} & \textbf{\phantom{no}gentle\phantom{n-}} \\ 
    \hline
    
    \textbf{\makecell[l]{Upper \\ Bound}} & \textbf{\makecell[c]{$\frac{d^2}{\epsilon^2}$}} (\citet{yu:LIPIcs.ITCS.2021.11}) \phantom{spaaaaaaaace} & \textbf{\makecell[c]{$\frac{d^{3}}{\epsilon^2 \alpha^2}$}} (Theorem~\ref{thm::upper_bound}) \\ 
    \hline
    
    \textbf{\makecell[l]{Lower \\ Bound}} & \textbf{\makecell[c]{$\frac{d^2}{\epsilon^2}$}} (\citet{liu2024rolesharedrandomnessquantum}) & \textbf{\makecell[c]{$\frac{d^{3}}{\epsilon^2 \alpha^2}$}} (Theorem~\ref{thm::lower_bound}) \\
\end{tabular}
\label{tab::bounds}
\vspace{10pt}
\caption{A comparison of the results on copy-complexity for quantum state certification in the gentle and non-gentle case.}
\end{table}

\section{Introduction to (quantum) statistics}

\subsection{Quantum states}

A $d$-dimensional quantum system is based on the Hilbert space $\mathbb{C}^d$ with the standard complex inner product. A state of the system is given by a positive self-adjoint matrix (positive for short) $\rho \in \mathbb{C}^{d \times d}$ with trace one, i.e. $\Tr[\rho] = 1$. A state $\rho$ of the system is said to be pure if it has rank one, i.e. $\rank(\rho) = 1$. In that case we can write $\rho = \ket{\psi}\bra{\psi}$ for a normalized vector $\ket{\psi} \in \mathbb{C}^d$. As such, we often identify pure states with one of their representational vectors $\ket{\psi}$. We denote by $\mathcal{S}(\mathbb{C}^d)$ the convex set of all quantum states and by $\mathcal{S}_{pure}(\mathbb{C}^d)$ the set of pure quantum states. Composite quantum systems are described by the tensor product of the individual parts. The product state of $n$ identical and independent states $\rho$ is given by $\rho^{\otimes n} := \rho \otimes ... \otimes \rho \in \mathcal{S}((\mathbb{C}^d)^{\otimes n})$.

\subsection{Quantum measurements} 
There are several notions of measurements in quantum mechanics that are more or less general. The one we chose here is as in \citet{NielsenChuang} and is the most general one that allows us to define the notion of post-measurement states which will be essential for gentleness. A quantum measurement is given by a set of measurement operators $M = (M_y)_{y \in \mathcal{Y}} \subseteq \mathbb{C}^{d \times d}$ such that
\begin{equation}
\label{eqn::resolution_of_identity}
    \sum_{y \in \mathcal{Y}} M_y^* M_y = \mathbbm{1},
\end{equation}
where $\mathbbm{1}$ denotes the $d \times d$ identity  and $M^*$ denotes the adjoint of a matrix $M$. The outcome of a quantum measurement is random and alters the state of the system. Measuring the states $\rho$ or $\ket{\psi}$ one obtains the outcome $y \in \mathcal{Y}$ with probability
\begin{align*}
    \mathbb{P}_{\rho}\left( R^M = y \right) = \Tr\left[ \rho M_y^*M_y \right] \hspace{20pt} &\text{and} \hspace{20pt} \mathbb{P}_{\ket{\psi}}\left( R^M = y \right) = \left| \bra{\psi}M_y^*M_y\ket{\psi} \right| = \norm{M_y \ket{\psi}}^2
\intertext{respectively. The state of the system after the measurement (the post-measurement state) is given by}
    \rho_{M \to y} = \frac{1}{\sqrt{\mathbb{P}_{\rho}\left( R^M = y \right)}} M_y\rho M_y^* \hspace{18pt} &\text{and} \hspace{18pt} \ket{\psi}_{M \to y} = \frac{1}{\sqrt{\mathbb{P}_{\ket{\psi}}\left( R^M = y \right)}} M_y \ket{\psi} = \frac{M_y \ket{\psi}}{\norm{M_y \ket{\psi}}}
\end{align*}
respectively. A measurement $M$ on a composite system $(\mathbb{C}^d)^{\otimes n}$ is said to be product if it acts independently on each state of the joint system and can be written as $M = M^{(1)} \otimes .... \otimes M^{(n)}$ given by measurement operators $M_y = M_{y_1}^{(1)} \otimes ... \otimes M_{y_n}^{(n)}$ for $y \in \mathcal{Y}_1 \times ... \times \mathcal{Y}_n$. If $M^{(1)} = ... = M^{(n)}$ we write $M = (M^{(1)})^{\otimes n}$ for short. The outcome of a product measurement on a product state are independent random variables. \medskip

A complete set of mutually unbiased bases (MUBs) is a collection $(|e_k^{(b)}\rangle)_{k = 1,...,d; b = 1,...,d+1}$ such that for each fixed $b$, the set $(|e_k^{(b)}\rangle)_{k = 1,...,d}$ is an orthonormal basis of $\mathbb{C}^d$ and for $b' \neq b$ it holds $|\langle e_k^{(b)}|e_{k'}^{(b')}\rangle|^2 = \frac{1}{d}$ for all $k, k' \in \{1,...,d\}$. The collection of all vectors of a complete set of mutually unbiased bases $(\ket{v_m})_{m =1,...,D}$ for $D = d(d+1)$ forms a quantum $2$-design (\citet{Klappenecker2005MutuallyUB}) and a such it holds (see \citet{liu2024rolesharedrandomnessquantum})
\begin{equation}
\label{eqn::property_2_design}
    \frac{1}{D} \sum_{m = 1}^D \bra{v_m}M\ket{v_m}^2 = \frac{1}{d(d+1)} \left( \Tr[M^2]  + \Tr[M]^2 \right)
\end{equation}
for a hermitian matrix $M$. Mutually unbiased bases have been proven to be optimal in both quantum state tomography (\citet{guta2018faststatetomographyoptimal}) and quantum state certification (\citet{liu2024rolesharedrandomnessquantum, yu:LIPIcs.ITCS.2021.11}). Since the basis-measurements defined by MUBs consist of rank-one operators, they cannot be gentle (\citet{butucea2025sampleoptimallearningquantumstates}). As such, in order to define a gentle measurement bases on MUBs for our state certification algorithm we must be more careful. Although the existence of a complete set of MUBs remains and open question in general dimensions, such a set can alwys be constructed for prime power dimensions $d = p^q$ which includes many interesting quantum system such as collections of lower-dimensional systems. 

\subsection{Quantum metrics}

The quantum analogues of the total variation and the $\mathbb{L}_2$ distances between likelihoods are the trace-norm distance and the Frobenius-norm distance, respectively, between quantum states. These distances between the states $\rho_1$ and $\rho_2$ are defined as
\begin{equation*}
    \norm{\rho_1 - \rho_2}_{Tr} = \frac{1}{2} \Tr\left[ |\rho_1 - \rho_2| \right] \hspace{20pt} \text{and} \hspace{20pt} \norm{\rho_1 - \rho_2}_F = \Tr\left[ (\rho_1 - \rho_2)^2 \right]^{\frac{1}{2}}
\end{equation*}
respectively. Here, for $A \in \mathbb{C}^{d \times d}$, $|A| = \sqrt{A^*A}$ denotes its matrix absolute value. Note that, since $\rho_1 - \rho_2$ is a self-adjoint operator, it admits a spectral decomposition $\rho_1 - \rho_2 = \sum_{j = 1}^d \mu_j \ket{\psi_j}\bra{\psi_j}$ for an orthonormal basis $(\ket{\psi}_j)_{j = 1,...,d}$ of $\mathbb{C}^d$. Then it holds
\begin{equation*}
    \norm{\rho_1 - \rho_2}_{Tr} = \frac{1}{2} \sum_{j = 1}^d |\mu_j| \hspace{20pt} \text{and} \hspace{20pt} \norm{\rho_1 - \rho_2}_F = \left(\sum_{j = 1}^d |\mu_j|^2\right)^{\frac{1}{2}},
\end{equation*}
i.e. the Trace-norm and Frobenius-norm are equivalent to the Schatten-1 and -2-norms respectively. As such, we also have
\begin{equation}
\label{eqn::Schatten_norm_inequalities}
    \norm{\rho_1 - \rho_2}_{Tr} \leq \frac{\sqrt{r}}{2} \norm{\rho_1 - \rho_2}_F,
\end{equation}
where $r = \rank(\rho_1 - \rho_2)$. Note that the Frobenius-norm is the norm induced by the inner product $\langle A, B \rangle = \Tr\left[ A^* B \right]$ on $\mathbb{C}^{d \times d}$. Both of the above norms have simpler forms for pure states.
\begin{lemma}{\normalfont (\citet{kargin2003chernoffboundefficiencyquantum})} \label{lemmaTraceNorm}
    Let $\rho_1 = \ket{\psi_1}\bra{\psi_1}, \rho_2 = \ket{\psi_2}\bra{\psi_2}$ be two pure states. Then the trace-norm and Frobenius-norm distance between the two is given by 
    \begin{equation*}
        \norm{\rho_1 - \rho_2}_{Tr} = \sqrt{1 - |\bra{\psi_1}\ket{\psi_2}|^2} \hspace{20pt} \text{and} \hspace{20pt} \norm{\rho_1 - \rho_2}_{F} = \sqrt{2(1 - |\bra{\psi_1}\ket{\psi_2}|^2)}.
    \end{equation*}
\end{lemma}

We can define the Trace- and Frobenius-norm for a super-operator $\mathcal{H}$ as the Schatten-1 and 2-norms of $\mathcal{H}$ on the Hilbert-space $\mathbb{C}^{d \times d}$ endowed with the inner product $\langle A, B \rangle = \Tr\left[ A^* B \right]$ on $\mathbb{C}^{d \times d}$. That is, for an orthonormal basis $(V_j)_{j = 1,...,d^2}$ of $\mathbb{C}^{d \times d}$ we define
\begin{equation*}
    \norm{\mathcal{H}}_{Tr} = \frac{1}{2} \sum_{j = 1}^{d^2} \left| \langle V_j, \mathcal{H}(V_j)\rangle \right|\hspace{20pt} \text{and} \hspace{20pt} \norm{\mathcal{H}}_F = \left( \sum_{j = 1}^{d^2} \left|\langle V_j, \mathcal{H}(V_j) \rangle \right|^2 \right)^{\frac{1}{2}}.
\end{equation*}

\subsection{Probability metrics}
In order to asses the quality of our tests we make use of metrics on the space of probability distributions. Notably the total-variation-distance that is intimately related to the optimal error of a test and the $\chi^2$-distance that allows us to bound the test error using properties that behave nicely for our set of local alternatives. The total-variation- and $\chi^2$-distances for mutually absolutely continuous probability distributions $\mathbb{P}_0, \mathbb{P}_1$ on $(\mathcal{Y}, \mathcal{P}(\mathcal{Y}))$ with probability mass functions $p$ and $q$ are define as 
\begin{equation*}
    \norm{\mathbb{P}_1 - \mathbb{P}_0}_{TV} = \frac{1}{2} \sum_{y \in \mathcal{Y}} \left| p_1(y) - p_0(y) \right| \hspace{20pt} \text{and} \hspace{20pt} d_{\chi^2}(\mathbb{P}_1, \mathbb{P}_0) = \sum_{y \in \mathcal{Y}} \frac{p_1(y)^2}{p_0(y)} -1
\end{equation*}
respectively. Let $H_0: \mathbb{P} = \mathbb{P}_0$ vs. $H_1: \mathbb{P} = \mathbb{P}_1$. Then for the error of the test $\Delta^*$ that distinguishes optimally between the two hypotheses $H_0$ and $H_1$, it holds
\begin{equation*}
    \mathbb{P}_0(\Delta^* = 1) + \mathbb{P}_1(\Delta^* = 0) = 1 - \norm{\mathbb{P}_0 - \mathbb{P}_1}_{TV}.
\end{equation*}
Together with the inequality $\norm{\mathbb{P}_1 - \mathbb{P}_0}_{TV} \leq \sqrt{d_{\chi^2}(\mathbb{P}_1, \mathbb{P}_0)}$ (see \citet{Tsybakov}), we can bound the optimal testing by controlling the $\chi^2$-distance. The following result will prove useful for our lower bounds as it evaluates the $\chi^2$-distance between a product distribution $\mathbb{P}_0=\mathbb{P}$ and an average $\mathbb{P}_1 = \mathbb{E}_{\nu}\left[ \mathbb{Q}_{\nu} \right]$ of product distributions $\mathbb{Q}_{\nu}$ with respect to a prior measure over the parameters $\nu$.
\begin{theorem}{(Lemma 8 in \citet{acharya_inference_2019})}
\label{thm::chi_squared_fluctuation}
    Let $\mathbb{P} = \mathbb{P}^{(1)} \otimes ... \otimes \mathbb{P}^{(n)}$ be a fixed probability distribution and $\mathbb{Q}_{\nu} = \mathbb{Q}_{\nu}^{(1)} \otimes ... \otimes \mathbb{Q}_{\nu}^{(n)}$ be another mutually absolutely continuous probability distribution indexed by a random parameter $\nu \in \mathcal{V}$ on a space $(\prod_{i = 1}^n \mathcal{Y}_i, \bigotimes_{i = 1}^n \mathscr{Y}_i)$ with densities $p(y) = \prod_{i = 1}^n p_i(y_i)$ and $q_{\nu}(y) = \prod_{i = 1}^n q_{\nu}^{(i)}(y_i)$ with respect to a common measure. Then it holds
    \begin{equation*}
        d_{\chi^2}\left( \mathbb{E}_{\nu}\left[ \mathbb{Q}_{\nu} \right], \mathbb{P} \right) = \mathbb{E}_{\nu_1, \nu_2} \left[ \prod_{i = 1}^n \left( 1 + H_i(\nu_1, \nu_2) \right) \right] - 1,
    \end{equation*}
    where $\nu_1, \nu_2$ are two independent copies of $\nu$ and
    \begin{equation*}
        H_i(\nu_1, \nu_2) = \mathbb{E}_{y_i \sim \mathbb{P}^{(i)}}\left[ \delta_{\nu_1}^{(i)}(y_i) \delta_{\nu_2}^{(i)}(y_i) \right] \hspace{10pt} \text{and} \hspace{10pt} \delta_{\nu}^{(i)}(y_i) = \frac{q_{\nu}^{(i)}(y_i) - p^{(i)}(y_i)}{p^{(i)}(y_i)}.
    \end{equation*}
\end{theorem}

\section{Gentle measurements}
\label{sec::gentleness}
Quantum measurements generally alter the state that they measure. The idea of gentleness is to limit the amount of alteration by the measurement. This renders the post-measurement states useful for further application. We achieve this by limiting the distance between the states before and after the measurements.

\begin{definition}
\label{defn::gentleness}
    For a given \textit{gentleness parameter} $\alpha \in [0,1]$, a measurement $M$ is $\alpha$-gentle on a set $\mathcal{S}$ of quantum states if for all possible measurement outcomes $y$
    \begin{equation*}
        \norm{\rho - \rho_{M \to y}}_{Tr} \leq \alpha \hspace{20pt} \text{for all } \rho \in \mathcal{S}.
    \end{equation*}
    
    If $\rho = \rho_1 \otimes ... \otimes \rho_n $ is a product state belonging to $ \mathcal{S}_1 \otimes ... \otimes \mathcal{S}_n =: \mathcal{S}^n$, we say that a measurement $M$ is locally-$\alpha$-gentle if it is a product measurement $M = M_1 \otimes ... \otimes M_n$ and $M_i$ is $\alpha$-gentle on $\mathcal{S}_i$ for all $i$.
\end{definition}

Note that the notion of locally $\alpha$-gentle product measurements is different to that of (possibly coherent) $\alpha$-gentle measurements $M$ on product states in $\mathcal{S}^n$. An example of the latter is given by the Laplace mechanism from \citet{aaronson2019gentle} which is not locally-$\alpha$-gentle.

A concept that has been shown (by \citet{aaronson2019gentle}) to be closely related to gentleness is that of quantum differential privacy. In contrast to gentleness, quantum differential privacy is only concerned with the probability of the measurement outcomes.

\begin{definition}
\label{defn::quantum_differential_privacy}
    A measurement $M $ is said to be $\delta$-quantum-differentially-private ($\delta$-qDP) for $\delta>0$ on a set $\mathcal{S}$ of quantum states, if for the outcome probabilities of any two states $\rho_1, \rho_2 $ in $\mathcal{S}$ under $M$ it holds
    \begin{equation}
    \label{eqn::quantum_privacy_equation}
        \mathbb{P}_{1}(R^{M} = y) \leq e^{\delta} \mathbb{P}_2(R^{M} = y) \hspace{20pt} \text{for all } y \in \mathcal{Y}.
    \end{equation}
    We call a product measurement $M = M_1 \otimes ... \otimes M_n$ locally-$\delta$-quantum-differentially private on $ \mathcal{S}_1 \otimes ... \otimes \mathcal{S}_n =: \mathcal{S}^n$ if each $M_i$ is $\delta$-quantum-differentially private on $\mathcal{S}_i$.
\end{definition}

We see that, as with gentleness, the locality of the quantum differential privacy for product measurements can be verified by checking the property on each register separately. One of the main results from \citet{aaronson2019gentle} was the fact that gentleness and quantum differential privacy represent two sides of the same coin. Note however that for any measurement $M = (M_y)_{y \in \mathcal{Y}}$ the quantum differential privacy of $M$ is only dependent on the operators $E_y = M_y^*M_y$ while the gentleness of $M$ is dependent on the operators $M_y$ themselves. Since for any family $U = (U_y)_{y \in \mathcal{Y}}$ of unitaries on $\mathbb{C}^d$, the operators $UM = (U_y M_y)_{y \in \mathcal{Y}}$ also define a quantum measurement with a possibly different gentleness value. However, since $(M_y U_y)^* (U_y M_y) = M_y^* M_y$, the outcome distributions of both measurements are the same, which means that their quantum differential privacy parameter is also the same. When talking about gentle measurements we must therefore always consider which implementation of $M$ we are working with. The following theorem gives an overview on the relation between the gentleness and quantum differential privacy. Its proof can be found in Appendix~\ref{sec::appendix_gentleness_proofs}.
\begin{theorem}
\label{thm::duality_gentleness_privacy}
    Let $M = (M_y)_{y \in \mathcal{Y}}$ be a quantum measurement. Then the following two results hold
    \begin{enumerate}
        \item[i)] If $M$ is locally-$\delta$-quantum-differentially-private on $\mathcal{S}(\mathbb{C}^d)^{\otimes n}$, then there exists an implementation of $M$ such that $M$ is $\alpha$-gentle on $\mathcal{S}(\mathbb{C}^d)^{\otimes n}$ for $\alpha = (e^{\frac{\delta}{2}} -1)/(e^{\frac{\delta}{2}} +1) = \tanh(\delta/4)$. The implementation is given by the positive definite operators $(|M_y|)_{y \in \mathcal{Y}}$.

        \item[ii)] If $M$ is locally-$\alpha$-gentle on $\mathcal{S}(\mathbb{C}^d)^{\otimes n}$ for $\alpha < 1/2$, then $M$ is locally-$\delta$-quantum-differentially-private on $\mathcal{S}(\mathbb{C}^d)^{\otimes n}$ for $\delta = 2\log((1+2\alpha)/(1-2\alpha)) = 4 \arctanh(2 \alpha)$.
    \end{enumerate}
\end{theorem}
\citet{butucea2025sampleoptimallearningquantumstates} have shown that the relation from gentleness to quantum differential privacy can be improved for positive-definite operators. For positive-definite operators, a locally-$\alpha$-gentle measurement is locally-$\delta$-quantum-differentially-private for $\delta = 2\log((1+\alpha)/(1-\alpha)) = 4 \arctanh(\alpha)$ (see Lemma~\ref{lem::improved_constant_positivity}). This result, together with Theorem~\ref{thm::duality_gentleness_privacy}, establishes a one-to-one correspondence between gentle and quantum differentially private measurements for positive-definite operators with optimal constants relating the two. \medskip

Another useful result for both gentleness as well as quantum differential privacy is the fact that it suffices to verify them on the set of pure quantum states. Proposition~\ref{prop::pure_states_imply_all_states} assures that it is enough to calculate the gentleness of a measurement on pure states to show gentleness on all states. Furthermore, Proposition~\ref{prop::equalities_qDP} assures that a measurement is locally-$\delta$-quantum-differentially-private on mixed states if and only if it is locally-$\delta$-quantum-differentially-private on pure states. Furthermore, in order to verify quantum differential privacy of a measurement $M = (M_y)_{y \in \mathcal{Y}}$ it suffices to check that $\lambda_{max}(E_{y}) \leq e^{\delta} \lambda_{min}(E_{y})$ for all $y \in \mathcal{Y}$, where $E_{y} = M_{y}^*M_{y}$.

\section{Quantum state certification}
\label{sec::upper_bound}
In quantum state certification we are given access to $n$ identical copies $\rho^{\otimes n}$ of an unknown state $\rho$. Our goal is to decide whether $\rho$ is equal to a known reference state $\rho_0$ or whether is $\epsilon$-far away from $\rho_0$ in trace-norm by performing a locally-$\alpha$-gentle measurement $M = M_1 \otimes ... \otimes M_n$ on $\rho^{\otimes n}$. Based on the random outcomes of the measurement, we then calculate a classical test $\Delta$ deciding for one of our two possible hypotheses. Formally, we are considering the hypothesis testing task
\begin{equation}
\label{eqn::testing_problem}
    H_0: \rho = \rho_0 \hspace{20pt} vs. \hspace{20pt} H_1: \norm{\rho - \rho_0}_{Tr}> \epsilon
\end{equation}
and the corresponding minimax testing error
\begin{equation*}
    P_e^* = \inf_{(M, \Delta)} \mathbb{P}_{\rho_0^{\otimes n}}^{R^M}\left( \Delta = 1 \right) + \sup_{\norm{\rho - \rho_0}_{Tr} > \epsilon} \mathbb{P}_{\rho^{\otimes n}}^{R^M}\left( \Delta = 0 \right),
\end{equation*}
where the infimum is taken over all locally-$\alpha$-gentle measurements $M$ and subsequent classical test functions $\Delta$. Crucially, once we have decided which measurement we perform, we are left with a classical testing task based on the random variables $R^{M_i} \sim \mathbb{P}_{\rho}^{R^{M_i}}$. The difficulty therefore lies in finding a measurement whose outcome distribution $p_{\rho}^{R^{M_i}}$ is closely related to - and most informative about - the quantum states themselves. In particular, we want $p_{\rho}^{R^{M_i}}$ to be very different from $p_{\rho_0}^{R^{M_i}}$ if $\rho$ is far away from $\rho_0$. We shall see that this is far from guaranteed even if we are not considering gentle measurements. Suppose that we chose the non-gentle basis measurement $M_i = M = (\ket{e_k}\bra{e_k})_{k = 1}^d$ for some orthonormal basis $(\ket{e_k})_{k = 1}^d$ of $\mathbb{C}^d$. For $\rho_0 = \rho_{mm} = \frac{1}{d} \mathbbm{1}$ being the maximally mixed state we have $p_{\rho_0}^{R^M}(k) = \frac{1}{d}$ for all $k \in \{1,...,d\}$, \emph{i.e.} the outcome distribution under $\rho_0$ is uniform. However, for the pure state $\rho = \ket{\psi}\bra{\psi}$ for
\begin{equation}
\label{eqn::pure_state_indistinguishable}
    \ket{\psi} = \sum_{k = 1}^d \frac{1}{\sqrt{d}} \ket{e_k}
\end{equation}
we have $p_{\rho}^{R^M}(k) = \frac{1}{d}$ as well even as $\norm{\rho - \rho_0}_{Tr} = 1 - 1/d$. As such, the outcome distributions of $M$ under $\rho$ and $\rho_0$ are identical meaning that this measurement cannot differentiate at all between $\rho$ and $\rho_0$. Note that this is true in the non-gentle case which is significantly easier than the gentle case. A way to circumvent this problem is by considering a complete set of mutually unbiased bases (MUBs) as was done in the first described quantum state certificatioin algorithm with fixed measurements by \citet{yu:LIPIcs.ITCS.2021.11}. MUBs are sets of orthonormal bases which are evenly spaced apart in the $d$-dimensional space $\mathbb{C}^d$. As such, there exists no state $\rho$ whose outcome distribution appears to be the same as the one of $\rho_0$ as there will always be one basis in which we can see a big difference in the distribution. 

In our case we must further deal with the additional constrain of gentleness. In general, measuring a state using a finite quantum measurement results in a multinomial distribution of the outcomes. Using the relation between gentleness and differential privacy it is natural to consider measurements that mirror the behavior of the optimal privacy kernels for multinomials like the RAPPOR mechanism \citet{acharya_inference_2021}. A measurement than gentle-izes a singular basis measurement runs into the same problems as the non-gentle version. There are pure states that have the same outcome distribution as the maximally mixed state. Such a measurement could never distinguish between these two states, which is only natural as its non-gentle counterpart cannot do so either. Another possible measurement would be one where we consider the optimal non-gentle measurement for state certification and gentle-ized it as a whole. However, such an optimal measurement, consisting of a complete set of MUBs, has $d(d+1) = O(d^2)$ measurement operators. The gentle version of such a measurement, that behaves like directly applying the RAPPOR mechanism on the outcomes, is however also suboptimal as it increases the factor of any subsequent test  by a factor of $d^2/\alpha^2$ compared to the non-gentle version. We see that constructing an optimal gentle measurement is more involved than applying the optimal privacy kernel to the optimal quantum measurement.

It turns out that the optimal way to perform gentle quantum state certification is by considering a complete set of MUBs but instead of gentle-izing the whole measurement, we gentle-ize each individual basis and combining them afterwards. The benefit of this way of measuring is that we retain the benefit of MUBs, which is that no state is simultaneously difficult to detect for every basis, while reducing the variance of the outcomes by a factor of $d$ as each basis only consists of $d$ elements.

\subsection{A locally gentle state certification algorithm}
Let us now demonstrate how to construct a locally-$\alpha$-gentle state certification algorithm. As we have already discussed, the construction of our measurement is based on mutually unbiased bases. We will assume that $d$ is such that there exists a complete set of $d+1$ mutually unbiased basis. This is always guaranteed when $d = p^q$ is a prime power, which includes most $d$-dimensional quantum systems of interest such as a system of $q$ joint qubits. The measurement we propose randomizes the basis-projections of the mutually unbiases bases in such a way that the output state remains close to the pre-measurement state. More precisely, let $D = d(d+1)$ and 
\begin{equation*}
    (\ket{v_m})_{m = 1,...,D} = \ket{e_k^{(b)}}_{\genfrac{}{}{0pt}{}{k = 1,...,d}{b = 1,...,d+1}}
\end{equation*}
be a complete set of mutually unbiased bases. Note that for every $m \in \{1,...,D\}$ there exist uniquely defined $k \in \{1,...,d\}$ and $b \in \{1,...,d+1\}$ such that $m = d(b-1) + k$. Now for every base $b$ and for every $z \in \{0,1\}^d$ we define the operators
\begin{equation}
\label{eqn::measurement_definion}
    M_{\delta, z}^{(b)} = \left( \frac{e^{\frac{\delta}{2}}}{e^{\frac{\delta}{2}} + 1} \right)^{\frac{d}{2}} \sum_{k = 1}^d e^{-\frac{\delta}{4} \norm{z - e_k}_1} \ket{e_k^{(b)}}\bra{e_k^{(b)}} \hspace{20pt} \text{and} \hspace{20pt} E_{\delta, z}^{(b)} = {M_{\delta, z}^{(b)}}^2,
\end{equation}
where $\delta = 4 \arctanh(\alpha)$ and $e_k \in \{0,1\}^d$ being the $k$-th standard basis vector, that is the vector that is zero everywhere except in the $k$-th entry. Lemma~\ref{lem::gentleness_of_gentleized_2_design} assures, that $M_{\delta}^{(b)} = (M_{\delta, z}^{(b)})_{z \in \{0,1\}^d}$ is in fact an $\alpha$-gentle quantum measurement. 

\begin{lemma}
\label{lem::gentleness_of_gentleized_2_design}
    Let $M_{\delta, z}^{(b)}$ be as in~\eqref{eqn::measurement_definion}. Then $M_{\delta}^{(b)} = (M_{\delta, z}^{(b)})_{z \in \{0,1\}^d}$ is an $\alpha$-gentle measurement on $\mathcal{S}(\mathbb{C}^d)$.
\end{lemma}

\begin{proof}
    We use Theorem~\ref{thm::duality_gentleness_privacy} to show the results. Since each $E_{\delta, z}^{(b)}$ is a sum of positive operators, $E_{\delta, z}^{(b)}$ is itself positive and we can take the operator square root $M_{\delta, z}^{(b)} = \sqrt{E_{\delta, z}^{(b)}}$. In order to assure that $M_{\delta}^{(b)} = (M_{\delta, z}^{(b)})_{z \in \{0,1\}^d}$ indeed defines a quantum measurement we must assure that it fulfills the completeness relation. We have
    \begin{align*}
        \sum_{z \in \{0,1\}^d} E_{\delta, z}^{(b)} &=  \left( \frac{e^{\frac{\delta}{2}}}{e^{\frac{\delta}{2}} + 1} \right)^d \sum_{k = 1}^d \ket{e_k^{(b)}}\bra{e_k^{(b)}} \sum_{z \in \{0,1\}^d} e^{-\frac{\delta}{2} \norm{z - e_k}_1} 
        \\
        &= \left( \frac{e^{\frac{\delta}{2}}}{e^{\frac{\delta}{2}} + 1} \right)^d \sum_{k = 1}^d \ket{e_k^{(b)}}\bra{e_k^{(b)}} \left[ \sum_{z | z_k = 1} e^{-\frac{\delta}{2} \norm{z - e_k}_1} + \sum_{z | z_k = 0} e^{-\frac{\delta}{2} \norm{z - e_k}_1}\right]
        \\
        &= \left( \frac{e^{\frac{\delta}{2}}}{e^{\frac{\delta}{2}} + 1} \right)^d \sum_{k = 1}^d \ket{e_k^{(b)}}\bra{e_k^{(b)}} \left[ \sum_{j = 0}^{d-1} \binom{d-1}{j} e^{-\frac{\delta}{2} j } + \sum_{j = 0}^{d-1} \binom{d-1}{j} e^{-\frac{\delta}{2} (j+1)} \right]
        \\
        &= \left( \frac{e^{\frac{\delta}{2}}}{e^{\frac{\delta}{2}} + 1} \right)^d \sum_{m = 1}^d \ket{e_k^{(b)}}\bra{e_k^{(b)}} \left( (e^{-\frac{\delta}{2}} + 1)^{d-1} + (e^{-\frac{\delta}{2}} +1)^{d-1} e^{-\frac{\delta}{2}} \right)
        \\
        &= \frac{e^{\frac{\delta}{2}}}{e^{\frac{\delta}{2}} + 1} \left( 1 + e^{-\frac{\delta}{2}} \right) \mathbbm{1}
        = \mathbbm{1}.
    \end{align*}
    Therefore, taking $M_{\delta, z}^{(b)} = \sqrt{E_{\delta, z}^{(b)}}$, the operators $M_{\delta}^{(b)} = (M_{\delta, z})_{z \in \{0,1\}^d}$ do indeed define a quantum measurement. Now, let $z \in \{0,1\}^d$ with $\norm{z} = l$ be fixed. We have
    \begin{equation*}
        e^{-\frac{\delta}{2} (l+1)} \leq e^{-\frac{\delta}{2} \norm{z - e_k}} \leq e^{-\frac{\delta}{2} (l-1)}
    \end{equation*}
    for all $l$ from which we obtain
    \begin{equation*}
        \frac{\mathbbm{P}_{\ket{\psi}}(R^{M_{\delta}^{(b)}} = z)}{\mathbbm{P}_{\ket{\psi'}}(R^{M_{\delta}^{(b)}} = z)} \leq e^{\frac{\delta}{2}((l+1)-(l-1))} \leq e^{\delta}
    \end{equation*}
    for any $\ket{\psi}, \ket{\psi'}$ which. Together with Theorem~\ref{thm::duality_gentleness_privacy} this shows the $\alpha$-gentleness of $M_{\delta}^{(b)}$.
\end{proof}

For each basis $b$, the measurement is a gentle-ized version of measuring in the basis $|e_k^{(b)}\rangle$. When measuring each of these measurements $n_b$ times we obtain $n_b$ independent and identically distributed random variables $R^{M_{\delta}^{(b)}}_i \in \{0,1\}^d$. If the non-gentle basis measurement had measured outcome $k$, outcome of the gentle version will be a vector $z \in \{0,1\}^d$ such that $z_k = 1$ and $z_l = 0$ for $l \neq k$ with high probability. In fact, for $\delta \to \infty$, this measurement coincides with its non-gentle counterpart when we identify $k$ with the vector $z = e_k$. Let us now describe how we can construct a state certification test based on the outcomes of the measurements $M_{\delta}^{(b)}$. First, we define
\begin{equation*}
    N_k^{(b)} = \sum_{i = 1}^{n_b} \mathbbm{1}_{\left\{ \left(R^{M_{\delta}^{(b)}}_i\right)_k = 1 \right\}} \in \mathbb{N}_0^d
\end{equation*}
which counts the occurrences of each entry being measured. We know that, correcting for some bias, this value is close to the non-gentle measurement outcome
\begin{equation}
\label{eqn::defn_p_rho}
    p_{\rho}^{(b)}(k) = \bra{e_k^{(b)}}\rho \ket{e_k^{(b)}}.
\end{equation}
In order to now test whether or not $\rho$ is equal to the reference state $\rho_0$, we can calculate the non-gentle outcome distribution $p_{\rho_0}^{(b)}$ as in~\eqref{eqn::defn_p_rho} and
compare it to the calculated values $N_k^{(b)}$ in the following way. Let $\beta = (e^{\frac{\delta}{2}} + 1)^{-1}$ and define
\begin{equation*}
    T_{n_b}^{(b)} = \sum_{k = 1}^d \left( (N_k^{(b)} - (n_b-1) (\alpha p_{\rho_0}^{(b)}(k) + \beta))^2 - N_k^{(b)} + (n_b-1)(\alpha p_{\rho_0}^{(b)}(k) + \beta)^2  \right).
\end{equation*}
$T_{n_b}^{(b)}$ will be small if $\rho = \rho_0$. If $\rho$ is far away from $\rho_0$, typically, the value of $T_{n_b}^{(b)}$ will be large. However, in the case $\rho$ is such that the basis $|e_k^{(b)}\rangle$ cannot distinguish well between $\rho$ and $\rho_0$, like we have described in~\eqref{eqn::pure_state_indistinguishable}, $T_{n_b}^{(b)}$ will also be small. In order to assure that we can always distinguish $\rho$ and $\rho_0$, we perform the measurements in every basis $b$, resulting in $d+1$ independent random variables $T_{n_b}^{(b)}$. We can then set $n_b = n/(d+1)$ for each $b$, meaning we measure each basis an equal amount of time, and define
\begin{equation*}
    T_n = \sum_{b = 1}^{d+1} T_{n_b}^{(b)}.
\end{equation*}
If $\rho = \rho_0$, $T_n$ will be small and if $\norm{\rho - \rho_0}_{Tr} > \epsilon$, at least one of the $T_{n_b}^{(b)}$ will be large, which results in a large $T_n$. We then define the test
\begin{equation}
\label{eqn::test_definition}
    \Delta_n = \mathbbm{1}_{\{T_n > c\}} \hspace{20pt} \text{ for } \hspace{20pt} c = \frac{1}{2} \frac{n_b(n_b-1)\alpha^2 \epsilon^2}{d},
\end{equation}
for which we have the following result.
\begin{theorem}
\label{thm::upper_bound}
    Consider the testing task~\eqref{eqn::testing_problem} and assume $n = O( \frac{d^3}{\epsilon^2 \alpha^2} )$. Then, for the test $\Delta_n$ defined in~\eqref{eqn::test_definition} it holds
    \begin{align}
        \mathbb{P}_{0}(\Delta_n = 1) + \sup_{\rho \,:\, \norm{\rho - \rho_0} > \epsilon} \mathbb{P}_{\rho}(\Delta_n = 0) \leq \frac{1}{3}.
    \end{align}
\end{theorem}

\begin{proof}
    In order to show the sample complexity of our algorithm, we relate the properties of the measurement~\eqref{eqn::measurement_definion} to the classical RAPPOR mechanism. The latter samples a $d$ dimensional vector based on an outcome $m \in \{1,...,d\}$ where it randomly flips the entries of the $m$-th basis vector $e_m$ with some probability according to~\eqref{eqn::privacy_kernel_multinomial}. To see this, let
    \begin{equation*}
        p_{\rho}^{(b)}(k) = \Tr\left[ \rho \ket{e_k^{(b)}}\bra{e_k^{(b)}} \right] = \bra{e_k^{(b)}}\rho \ket{e_k^{(b)}}
    \end{equation*}
    be the probability of measuring the state $\rho$ non-gently using in the basis $(|e_k^{(b)}\rangle)_{k = 1}^d$ directly (that is, using the measurement $M^{(b)} = (|e_k^{(b)}\rangle\langle e_k^{(b)}|)_{k = 1}^d$). The crucial result of our proof is now that the outcome probabilities of the measurement~\eqref{eqn::measurement_definion} are equal to the outcomes of applying the RAPPOR mechanism to the probability distribution $p_{\rho}^{(b)}$. To see this equality, consider the privatization mechanism described by, where, depending on the outputs of $R^{M^{(b)}}$, we independently draw $d$ elements $z_j \in \{0,1\}$ according to
    \begin{align}
        Q_j(Z_j = z_j | R^{M^{(b)}} = k)
        &= \frac{1}{e^{\frac{\delta}{2}}+1} \begin{cases}
            z_j e^{\frac{\delta}{2}} + (1-z_j) \hspace{10pt}& j = k
            \\
            (1-z_j) e^{\frac{\delta}{2}} + z_j \hspace{10pt}& j \neq k
        \end{cases}.\label{eqn::privacy_kernel_multinomial}
    \end{align}
    We let $Z = (Z_1,...,Z_d)$. Then
    \begin{align*}
        Q(Z = z |R^{M^{(b)}} = k) &= \prod_{j = 1}^d Q_j(Z_j = z_j | R^{M^{(b)}} = k)
        \\
        &= \left( \frac{1}{e^{\frac{\delta}{2}} +1} \right)^d \hspace*{-3pt} Q_k(Z_k = z_k | R^{M^{(b)}} = k) \prod_{j \neq k}^d Q_j(Z_j = z_j | R^{M^{(b)}} = k) 
        \\
        &= \left( \frac{1}{e^{\frac{\delta}{2}} +1} \right)^d  \hspace*{-3pt} \left( e^{\frac{\delta}{2}} \right)^{(z - e_k)_k} \prod_{j \neq k} \left( e^{\frac{\delta}{2}} \right)^{(z - e_k)_j}
        \hspace*{-5pt} = \left( \frac{e^{\frac{\delta}{2}}}{e^{\frac{\delta}{2}} +1} \right)^d \hspace*{-5pt} e^{-\frac{\delta}{2} \norm{z - e_k}_1}.
    \end{align*}
    This shows
    \begin{align*}
        \mathbb{P}_{\rho}(Z = z) &= \sum_{m = 1}^d Q(Z = z | R^{M^{(b)}} = k) \mathbb{P}_{\rho}(R^{M^{(b)}} = k)
        \\
        &= \sum_{k = 1}^d \left( \frac{e^{\frac{\delta}{2}}}{e^{\frac{\delta}{2}} +1} \right)^d e^{-\frac{\delta}{2} \norm{z - e_k}_1} \Tr\left[ \rho \ket{e_k^{(b)}}\bra{e_k^{(b)}} \right]
        \\
        &= \Tr\left[ \left( \frac{e^{\frac{\delta}{2}}}{e^{\frac{\delta}{2}} +1} \right)^d \sum_{k = 1}^d e^{-\frac{\delta}{2} \norm{z - e_k}_1} \ket{e_k^{(b)}}\bra{e_k^{(b)}} \rho \right] = \mathbb{P}_{\rho}(M_{\delta}^{(b)} = z).
    \end{align*}
    Having established this equality allows us to now apply results from differential privacy. Any statistic based on the results of our gentle measurement~\eqref{eqn::measurement_definion} has the same properties as applied to the privatized version of the outcome distribution of the non-gentle basis measurement. In particular, applying \citet{acharya_inference_2021} Lemma III.3 and III.4 to $T_n$, we obtain
    \begin{equation*}
        \mathbb{E}_{\rho}\left[T_n^{(b)}\right] = n_b (n_b-1) \alpha^2 \norm{p_{\rho}^{(b)} - p_{\rho_0}^{(b)}}_2^2 
    \end{equation*}
    and 
    \begin{equation*}
        \Var_{\rho}\left[T_n^{(b)}\right] \leq 2 d n_b^2 + 5 n_b^3 \alpha^2 \norm{p_{\rho}^{(b)} - p_{\rho_0}^{(b)}}_2^2 = 2dn_b^2 + 5n_b\mathbb{E}_{\rho}\left[T_n^{(b)}\right].
    \end{equation*}
    With these equations we can calculate
    \begin{align*}
        \mathbb{E}_{\rho}\left[ T_n \right] = \sum_{b = 1}^{d+1} \mathbb{E}_{\rho}\left[ T_n^{(b)} \right] &= n_b (n_b-1) \alpha^2 \sum_{b = 1}^{d+1} \norm{p_{\rho}^{(b)} - p_{\rho_0}^{(b)}}_2^2 
        \\
        &= n_b (n_b - 1) \alpha^2 \sum_{b = 1}^{d+1} \sum_{k = 1}^d \bra{e_k^{(b)}}(\rho - \rho_0)\ket{e_k^{(b)}}^2
        \\
        &= n_b (n_b - 1) \alpha^2 \sum_{m = 1}^D \bra{v_m}(\rho - \rho_0)\ket{v_m}^2 = n_b (n_b-1) \alpha^2 \norm{\rho - \rho_0}_F^2,
    \end{align*}
    where we used the $2$-design property~\eqref{eqn::property_2_design} of the complete set of mutually unbiased bases
    \begin{equation*}
        \left(\ket{e_k^{(b)}}\right)_{\genfrac{}{}{0pt}{}{k = 1,...,d}{b = 1,...,d+1}} = \left( \ket{v_m} \right)_{m = 1,...,D}.
    \end{equation*}
    The same argument allows for the calculation of the variance as
    \begin{equation*}
        \Var_{\rho}\left[ T_n \right] = \sum_{b = 1}^{d+1} \Var_{\rho}\left[T_n^{(b)}\right] \leq 2 d(d+1) n_b^2 + 5n_b^3 \alpha^2 \norm{\rho - \rho_0}_F^2 \leq 2d(d+1) n_b^2 + 5n_b \mathbb{E}_{\rho}\left[ T_n \right],
    \end{equation*}
    Now, under the null hypothesis we have $\mathbb{E}_{\rho}[T_n] = 0$ and therefore
    \begin{align*}
        \mathbb{P}_{\rho_0}\left( \Delta_n = 1 \right) = \mathbb{P}_{\rho_0}\left( T_n > c \right) \leq \frac{1}{c^2} \Var_{\rho_0}\left[T_n \right] \leq \frac{8r^2}{n_b^4 \alpha^4 \epsilon^4} d^2 n_b^2 = \frac{8 d^6}{n^2 \alpha^4 \epsilon^4},
    \end{align*}
    where we note that $n_b = n/(d+1)$. Under the alternative, using inequality~\eqref{eqn::Schatten_norm_inequalities}, we have 
    $$
    \norm{\rho - \rho_0}_F^2 \geq \frac{1}{d} \norm{\rho - \rho_0}_{Tr}^2 \geq \frac{\epsilon^2}d$$
    and with $n_b = n/(d+1)$ we therefore have
    \begin{align*}
        \mathbb{P}_{\rho}\left( \Delta_n = 0 \right) 
        \leq \mathbb{P}_{\rho}\left( T_n \leq \frac{1}{2} \mathbb{E}_{\rho}\left[ T_n \right] \right) 
        &\leq \frac{2d^2n_b^2 d^2}{n_b^4 \alpha^4 \epsilon^4} + 5 \frac{n_b d}{n_b^2 \alpha^2 \epsilon^2}
        \leq \frac{2 d^6}{n^2 \alpha^4 \epsilon^4} + 5 \frac{d^2}{n \alpha^2 \epsilon^2}.
    \end{align*}
    For the sum of error we then have
    \begin{align*}
        \mathbb{P}_{\rho_0}\left( \Delta_n = 1 \right) + \sup_{\norm{\rho - \rho_0}_{Tr} > \epsilon} \mathbb{P}_{\rho}\left( \Delta_n = 0 \right) &\leq \frac{8 d^6}{n^2 \alpha^4 \epsilon^4} +  \frac{2 d^6}{n^2 \alpha^4 \epsilon^4} + 5 \frac{d^2}{n \alpha^2 \epsilon^2} = 10 \frac{d^6}{n^2 \alpha^4 \epsilon^2} + 5 \frac{d^2}{n\alpha^2 \epsilon^2}.
    \end{align*}
    We see that the last term is smaller than $\frac{1}{3}$ for $n = O(\frac{d^3}{\epsilon^2 \alpha^2})$.
\end{proof}

\begin{remark}
    When we additionally know that $\rho_0$ and $\rho$ are rank $r$ quantum states we can make use of the inequality $\norm{\rho - \rho_0}_F^2 \geq \norm{\rho - \rho_0}_{Tr}^2 /(\max\{2r,d\})\geq \epsilon^2/(\max\{2r,d\})$ to show the improved sample complexity $n = O(\frac{d^2r}{\epsilon^2 \alpha^2})$.
\end{remark}

\section{Optimality of the algorithm}
Let us now demonstrate that the algorithm in Section~\ref{sec::upper_bound} is in fact sample optimal for locally-$\alpha$-gentle quantum state certification. 

\begin{theorem}
\label{thm::lower_bound}
Consider the testing task~\eqref{eqn::testing_problem} for $\rho_0 = \frac{1}{d}\mathbbm{1}$ and assume $n = O(\frac{d^3}{\epsilon^2 \alpha^2})$. Then, for the minimal testing error it holds
\begin{equation}
\label{eqn::test_error}
    \inf_{(M, \Delta)} \mathbb{P}_{\rho_0^{\otimes n}}^{R^{M}}\left( \Delta = 1 \right) + \sup_{\rho: \norm{\rho - \rho_0}_{Tr} > \epsilon} \mathbb{P}_{\rho^{\otimes n}}^{R^{M}}\left( \Delta = 0 \right) \geq \frac{1}{3}.
\end{equation}
Here the infimum is taken over over all locally-$\alpha$-gentle measurements $M$ and subsequent tests $\Delta$.
\end{theorem}

Let us first recall the following result from \citet{liu2024rolesharedrandomnessquantum} which we will prove again for the readers convenience.

\begin{lemma}[Proposition 3 in \citet{liu2024rolesharedrandomnessquantum}]
\label{lem::probability_valid_states}
    Let $d^2/2 \leq D \leq d^2-1$ and $\nu$ be drawn uniformly from $\mathcal{V} = \{-1,1\}^D$. Let $\Delta_{\nu}$ and $\sigma_{\nu}$ as in~\eqref{eqn::defn_alternative_states_2}. Then there exists a universal constant $c \leq 10\sqrt{2}$ such that for $\epsilon < 1/c^2$ we have
    \begin{equation*}
        \mathbb{P}_{\nu}\left( \norm{\Delta_{\nu}}_{op} \leq \frac{1}{d} \text{ and } \norm{\rho_{\nu} - \rho_{0}} \geq \epsilon \right) \geq 1 - 2e^{-d}.
    \end{equation*}
\end{lemma}
\begin{proof}
    The Hölder inequality for matrices gives
    \begin{equation*}
        \norm{\Delta_{\nu}}_{F}^2 \leq \norm{\Delta_{\nu}}_{op} \norm{\Delta_{\nu}}_{1}.
    \end{equation*}
    Applying Theorem 15 in \citet{liu2024rolesharedrandomnessquantum} to $\Delta_{\nu} = \frac{c\epsilon}{\sqrt{Dd}} W$ with $\norm{\Delta_{\nu}}_F = \frac{c \epsilon}{\sqrt{d}}$ gives
    \begin{equation*}
        \mathbb{P}_{\nu}\left( \norm{\Delta_{\nu}}_{op} \leq \frac{c \epsilon}{\sqrt{Dd}} 10 \sqrt{d}\right) = \mathbb{P}_{\nu}\left( \norm{W}_{op} \leq 10 \sqrt{d} \right) \geq 1 - 2e^{-d}.
    \end{equation*}
    As such, we also have
    \begin{align*}
        \mathbb{P}_{\nu} \left( \norm{\Delta_{\nu}}_1 \geq \frac{c\epsilon \sqrt{D}}{10d} \right) \geq \mathbb{P}_{\nu}\left( \norm{\Delta_{\nu}}_{op} \leq \norm{\Delta_{\nu}}_F^2 \frac{10d}{c\epsilon\sqrt{D}} \right) = \mathbb{P}_{\nu}\left( \norm{W}_{op} \leq 10 \sqrt{d} \right) \geq 1 - 2e^{-d}.
    \end{align*}
    Since $D \geq \frac{d^2}{2}$, for $c = 10\sqrt{2}$, for all $\epsilon< \frac{1}{c^2}$ we have
    \begin{equation*}
        \mathbb{P}_{\nu}\left( \norm{\Delta_{\nu}}_{op} \leq \frac{1}{d} \text{ and } \norm{\rho_{\nu} - \rho_{0}} \geq \epsilon \right) \geq 1 - 2e^{-d}.
    \end{equation*}
\end{proof}

With the result of the Lemma~\ref{lem::probability_valid_states} we can now turn to the proof of Theorem~\ref{thm::lower_bound}.

\begin{proof}[Proof of Theorem~\ref{thm::lower_bound}]
     Suppose we are given a $n$-copies of the unknown state $\rho$, that is $\rho^{\otimes n}$ on which we perform a locally-$\alpha$-gentle measurement $M^{\otimes n}$. The error of any subsequent test in then given by
    \begin{equation*}
        \mathbb{P}_{\rho_0^{\otimes n}}^{R^{M^{\otimes n}}}\left( \Delta = 1 \right) + \sup_{\rho: \norm{\rho - \rho_0}_{Tr} > \epsilon} \mathbb{P}_{\rho^{\otimes n}}^{R^{M^{\otimes n}}}\left( \Delta = 0 \right)
    \end{equation*}
    We will reduce the supremum over the alternative to the maximum over a well suited finite set of alternatives. Let us denote by $\mathbb{H}_d$ the real vector space of $d \times d$ hermitian matrices with inner product given by $\langle A , B \rangle = \Tr\left[ A^* B \right]$. Let $(V_j)_{j = 1}^{d^2}$ be an orthonormal basis of $\mathbb{H}_d$ with $V_{d^2} = \frac{1}{\sqrt{d}} \mathbbm{1}$. The exact choice of the $V_j$ will depend on the measurement and will be given later in the proof. For some $d^2/2 \leq D \leq d^2 - 1$ we will now define the states $\rho_{\nu}$ by
    \begin{equation}
    \label{eqn::defn_alternative_states_2}
        \rho_{\nu} = \rho_0 + \Delta_{\nu} = \rho_0 +  \frac{c \epsilon}{\sqrt{dD}} \sum_{i = 1}^D \nu_i V_i 
    \end{equation}
    for $\nu \in \mathcal{V} = \{-1,1\}^D$. Let us now consider $\mathcal{V}_{\mathcal{S}} := \left\{ \nu \in \mathcal{V} \middle| \rho_{\nu} \in \mathcal{S}(\mathbb{C}) \text{ and } \norm{\rho_{\nu} - \rho_0}_{\Tr} > \epsilon \right\}$. The set $\mathcal{V}_{\mathcal{S}}$ indexes those operators $\rho_{\nu}$ that actually define quantum states (in that they are positive) that are at least $\epsilon$ far away from $\rho_0$ in trace-norm. Lemma~\ref{lem::probability_valid_states} assures us that for a suitable $c > 0$, we have $|\mathcal{V}_{\mathcal{S}}|/|\mathcal{V}| \geq 1 - 2e^{-d}$ which we will later use to work with the whole set $\mathcal{V}$. We can now lower bound the testing error as follows
    \begin{align*}
        \mathbb{P}_{\rho_0^{\otimes n}}^{R^{M}}\left( \Delta = 1 \right) + \sup_{\rho: \norm{\rho - \rho_0}_{Tr} > \epsilon} \mathbb{P}_{\rho^{\otimes n}}^{R^{M}}\left( \Delta = 0 \right) 
        &\geq \mathbb{P}_{\rho_0^{\otimes n}}^{R^{M}}\left( \Delta = 1 \right) + \max_{\nu \in \mathcal{V}_{\mathcal{S}}} \mathbb{P}_{\rho_{\nu}^{\otimes n}}^{R^{M}}\left( \Delta = 0 \right)
        \\
        &\geq \mathbb{P}_{\rho_0^{\otimes n}}^{R^{M}}\left( \Delta = 1 \right) + \frac{1}{|\mathcal{V}_{\mathcal{S}}|} \sum_{\nu \in \mathcal{V}_{\mathcal{S}}} \mathbb{P}_{\rho_{\nu}^{\otimes n}}^{R^{M}}\left( \Delta = 0 \right)
        \\
        &\geq \mathbb{P}_{\rho_0^{\otimes n}}^{R^{M}}\left( \Delta = 1 \right) + \mathbb{E}_{\nu} \left[ \mathbb{P}_{\rho_{\nu}^{\otimes n}}^{R^{M}} \right] \left( \Delta = 0 \right)
        \\
        &\geq 1 - \sqrt{d_{\chi^2}\left( \mathbb{P}_{\rho_0^{\otimes n}}^{R^{M}}, \mathbb{E}_{\nu \sim \mathcal{V}_{\mathcal{S}} } \left[ \mathbb{P}_{\rho_{\nu}^{\otimes n}}^{R^{M}} \right] \right)}. 
    \end{align*}
    In order to bound this we use the decoupled $\chi^2$ fluctuation which we can calculate using Theorem \ref{thm::chi_squared_fluctuation}. As such, we need to calculate
    \begin{equation*}
        H_i(\nu_1, \nu_2) = \mathbb{E}_{y_i \sim \mathbb{P}^{(i)}}\left[ \delta_{\nu_1}^{(i)}(y_i) \delta_{\nu_2}^{(i)}(y_i) \right] \hspace{10pt} \text{and} \hspace{10pt} \delta_{\nu}^{(i)}(y_i) = \frac{q_{\nu}^{(i)}(y_i) - p^{(i)}(y_i)}{p^{(i)}(y_i)}.
    \end{equation*}
    In our particular setup, we have
    \begin{equation}
    \label{eqn::defn_p_and_q}
        p^{(i)}(y_i) = \Tr(\rho_{0} E_{y_i}) \hspace{20pt} \text{and} \hspace{20pt} q_{\nu}^{(i)}(y_i) = \Tr(\rho_{\nu} E_{y_i}),
    \end{equation}
    where $\nu$ is drawn uniform at random from $\mathcal{V}_{\mathcal{S}}$ and $\rho_{mm}$ and $\rho_{\nu} = \rho_{mm} + \Delta_{\nu}$. This gives 
    \begin{align}
        H_i(\nu_1, \nu_2) 
        &= d \sum_{y_i \in \mathcal{Y}_i} \frac{\Tr\left[ \Delta_{\nu_1} E_{y_i} \right] \Tr\left[ \Delta_{\nu_2} E_{y_i} \right]}{\Tr[E_{y_i}]} \notag
        \\
        &= d \Tr\left[ \Delta_{\nu_1} \sum_{{y_i} \in \mathcal{Y}_i} \Tr\left[ \Delta_{\nu_2} E_{y_i} \right] \frac{1}{\Tr[E_{y_i}]} E_{y_i} \right] = d \Tr\left[ \Delta_{\nu_1} \mathcal{H}_i(\Delta_{\nu_2}) \right], \label{eqn:relation_H_i_Lüders_channel}
    \end{align}
    where
    \begin{align}
    \label{eqn::Kraus_channel_single}
        \mathcal{H}_i: \mathbb{C}^{d \times d} \to \mathbb{C}^{d \times d} \hspace{10pt} \text{with}& \hspace{10pt} \mathcal{H}_i(A) := \sum_{{y_i} \in \mathcal{Y}_i} \Tr\left[ A E_{y_i} \right] \frac{1}{\Tr[E_{y_i}]} E_{y_i}
        \intertext{and}
        \label{eqn::average_Kraus_channel}
        &\bar{\mathcal{H}} = \frac{1}{n} \sum_{i = 1}^n \mathcal{H}_i
    \end{align}
    are linear super-operators. We will now show that these super-operators have several properties which will prove useful.
    
    \begin{lemma}
\label{lem::properties_of_Kraus_channel}
    Let $\mathcal{H}_i$ and $\bar{\mathcal{H}}$ be defined as in~\eqref{eqn::Kraus_channel_single} and~\eqref{eqn::average_Kraus_channel}. Then
    \begin{enumerate}
        \item[(i)] $\mathcal{H}_i$ and $\bar{\mathcal{H}}$ are hermitian and positive.

        \item[(ii)] $\mathcal{H}_i$ and $\bar{\mathcal{H}}$ are Hermiticity-preserving.

        \item[(iii)] $\mathcal{H}_i$ and $\bar{\mathcal{H}}$ are Trace-preserving.

        \item[(iv)] $\mathcal{H}_i$ and $\bar{\mathcal{H}}$ are unital.
    \end{enumerate}
\end{lemma}

\begin{proof}
    \begin{enumerate}
        \item[(i)] Let $A \in \mathbb{C}^{d \times d}$. Then
        \begin{align*}
            \Tr\left[ A^* \mathcal{H}_i(A) \right] &= \Tr\left[ A^* \sum_{{y_i} \in \mathcal{Y}_i} \Tr\left[ A E_{y_i} \right] \frac{1}{\Tr[E_{y_i}]} E_{y_i} \right]
            \\
            &= \sum_{{y_i} \in \mathcal{Y}_i} \Tr\left[ A^* E_{y_i} \right] \Tr\left[ A E_{y_i} \right] \frac{1}{\Tr[E_{y_i}]}
            \\
            &= \sum_{{y_i} \in \mathcal{Y}_i} |\Tr\left[ A E_{y_i} \right]|^2 \frac{1}{\Tr[E_{y_i}]} \geq 0
        \end{align*}
        using the fact that $\Tr[A^*E_{y_i}] = \Tr\left[(E_{y_i} A)^*\right] = \overline{\Tr\left[E_{y_i} A\right]}$. As such $\mathcal{H}_i$ is hermitian and positive and with it $\bar{\mathcal{H}}$.

        \item[(ii)] Let $A \in \mathbb{H}_d$ be hermitian. Then
        \begin{align*}
            \mathcal{H}_i(A)^* &= \left(\sum_{{y_i} \in \mathcal{Y}_i} \Tr\left[ A E_{y_i} \right] \frac{1}{\Tr[E_{y_i}]} E_{y_i}\right)^* 
            \\
            &= \sum_{{y_i} \in \mathcal{Y}_i} \Tr\left[ A^* E_{y_i} \right] \frac{1}{\Tr[E_{y_i}]} E_{y_i} = \mathcal{H}_i(A^*) = \mathcal{H}_i(A).
        \end{align*}
        By linearity, $\bar{\mathcal{H}}(A)^* = \bar{\mathcal{H}}(A)$.

        \item[(iii)] Let $A \in \mathbb{C}^{d \times d}$. Then
        \begin{align*}
            \Tr\left[ \mathcal{H}_i(A) \right] &= \Tr\left[ \sum_{{y_i} \in \mathcal{Y}_i} \Tr\left[ A E_{y_i} \right] \frac{1}{\Tr[E_{y_i}]} E_{y_i} \right] 
            \\
            &= \sum_{{y_i} \in \mathcal{Y}_i} \Tr\left[ A E_{y_i} \right] = \Tr\left[ A \sum_{{y_i} \in \mathcal{Y}_i} E_{y_i} \right] = \Tr[A].
        \end{align*}
        By linearity, $\bar{\mathcal{H}}$ is also trace-preserving.

        \item[(iv)] We have
        \begin{equation*}
            \mathcal{H}_i(\mathbbm{1}) = \sum_{{y_i} \in \mathcal{Y}_i} \Tr\left[ \mathbbm{1} E_{y_i} \right] \frac{1}{\Tr[E_{y_i}]} E_{y_i} = \sum_{{y_i} \in \mathcal{Y}_i} E_{y_i} = \mathbbm{1}.
        \end{equation*}
        By linearity, $\bar{\mathcal{H}}$ is also unital.
    \end{enumerate}
\end{proof}
    \noindent Let us now continue with the proof of Theorem~\ref{thm::lower_bound}. Using the linear super-operator form, we can rewrite the decoupled $\chi^2$ fluctuation as 
    \begin{align*}
        d_{\chi^2}\left( \mathbb{P}_{\rho_0^{\otimes n}}^{R^{M}}, \mathbb{E}_{\nu} \left[ \mathbb{P}_{\rho_{\nu}^{\otimes n}}^{R^{M}} \right] \right) &= \mathbb{E}_{\nu_1, \nu_2 \sim U(\mathcal{V}_{\mathcal{S}})} \left[ \prod_{i = 1}^n \left( 1 + H_i(\nu_1, \nu_2) \right) \right] - 1
        \\
        &= \mathbb{E}_{\nu_1, \nu_2 \sim U(\mathcal{V}_{\mathcal{S}})} \left[\exp\left(d \sum_{i = 1}^n \Tr\left[ \Delta_{\nu_1} \mathcal{H}_i(\Delta_{\nu_2}) \right] \right) \right] - 1
        \\
        &= \mathbb{E}_{\nu_1, \nu_2 \sim U(\mathcal{V}_{\mathcal{S}})} \left[ \exp\left(nd \langle \Delta_{\nu_1},  \bar{\mathcal{H}}(\Delta_{\nu_2}) \rangle \right) \right] - 1.
    \end{align*}
    Now, we use the fact that the probability of an alternative state not being a valid quantum states is exponentially small (see Lemma~\ref{lem::probability_valid_states}) in order to further bound
    \begin{equation*}
        \mathbb{E}_{\nu_1, \nu_2 \sim U(\mathcal{V}_{\mathcal{S}})} \left[ \exp\left(nd \langle \Delta_{\nu_1},  \bar{\mathcal{H}}(\Delta_{\nu_2}) \rangle \right) \right] \leq \left( \frac{e^d}{e^d -2} \right)^2 \mathbb{E}_{\nu_1, \nu_2 \sim U(\mathcal{V})} \left[ \exp\left(nd \langle \Delta_{\nu_1},  \bar{\mathcal{H}}(\Delta_{\nu_2}) \rangle \right) \right]
    \end{equation*}
    where we will write $\mathbb{E}_{\nu_1, \nu_2} := \mathbb{E}_{\nu_1, \nu_2 \sim U(\mathcal{V})}$ in short from now on. For $D \leq d^2 - 1$, let $\mathcal{V}_D = \left( V_i \right)_{i = 1,...,D}$, where the $V_i$ are the eigenvectors of $\Bar{\mathcal{H}}$ (Note that by Lemma~\ref{lem::properties_of_Kraus_channel} and the spectral theorem such an orthonormal basis of eigenvectors/eigenmatrices of $\bar{\mathcal{H}}$ always exists). Then, for $\nu \in \{-1,1\}^D$, we have
    \begin{equation*}
        \Delta_{\nu} =  \frac{c \epsilon}{\sqrt{dD}} \sum_{i = 1}^D \nu_i V_i \in \linspan (\mathcal{V}_D).
    \end{equation*}
    Let $\Phi_{\mathcal{V}_D}^{-1}(\Delta_\nu) = \left( \frac{c \epsilon}{\sqrt{dD}} \nu_i \right)_i = \frac{c \epsilon}{\sqrt{dD}} \nu$ be the coefficient vector of $\Delta_{\nu}$ with respect to the basis $\mathcal{V}_D$ of $\linspan (\mathcal{V}_D)$. Furthermore, let $M_D =\mathcal{M}_{\mathcal{V}_D}^{\mathcal{V}_D}(\Bar{\mathcal{H}}|_{\linspan (\mathcal{V}_D)}) = \diag\left( \mu_1,...,\mu_D \right)$ be the transformation matrix of $\Bar{\mathcal{H}}|_{\linspan( \mathcal{V}_D)}$ with respect to the basis $\mathcal{V}_D$. Then it holds
    \begin{align*}
        \langle \Delta_{\nu_1}, \bar{\mathcal{H}}(\Delta_{\nu_2}) \rangle &= \left(\Phi_{\mathcal{V}_D}^{-1}(\Delta_{\nu_1})\right)^t \mathcal{M}_{\mathcal{V}_D}^{\mathcal{V}_D}(\Bar{\mathcal{H}}|_{\linspan( \mathcal{V}_D )}) \Phi_{\mathcal{V}_D}^{-1}(\Delta_{\nu_2}) 
        \\
        &= \frac{c^2 \epsilon^2}{dD} {\nu_1}^t M_D \nu_2 
        \\
        &= \frac{c^2 \epsilon^2}{dD} \sum_{i = 1}^D \mu_i \nu_{1,i} \nu_{2,i}.
    \end{align*}
    This now allows us to further bound the $\chi^2$ distance as.
    \begin{align*}
        d_{\chi^2}\left( \mathbb{P}_{\rho_0^{\otimes n}}^{R^{M}}, \mathbb{E}_{\nu} \left[ \mathbb{P}_{\rho_{\nu}^{\otimes n}}^{R^{M}} \right] \right) &\leq \left( \frac{e^d}{e^d -2} \right)^2 \mathbb{E}_{\nu_1, \nu_2} \left[ \exp\left(nd \langle \Delta_{\nu_1},  \bar{\mathcal{H}}(\Delta_{\nu_2}) \rangle \right) \right] - 1
        \\
        &= \left( \frac{e^d}{e^d -2} \right)^2 \mathbb{E}_{\nu_1, \nu_2} \left[ \exp\left(\frac{d c^2 \epsilon^2}{D} \sum_{i = 1}^D \mu_i \nu_{1,i} \nu_{2,i}\right) \right] - 1.
    \end{align*}
    Using the tower property of the conditional expectation, we write 
    \begin{align*}
        \mathbb{E}_{\nu_1, \nu_2} \left[ \exp\left(\frac{d c^2 \epsilon^2}{D} \sum_{i = 1}^D \mu_i \nu_{1,i} \nu_{2,i}\right) \right] &= \mathbb{E}_{\nu_1} \left[ \mathbb{E}_{\nu_2}\left[ \exp\left(\frac{d c^2 \epsilon^2}{D} \sum_{i = 1}^D \mu_i \nu_{1,i} \nu_{2,i}\right) \middle| \nu_1 \right] \right]
        \\
        &\leq \mathbb{E}_{\nu_1}\left[ \exp\left(\frac{1}{2} \frac{n^2 c^4 \epsilon^4}{D^2} \sum_{i = 1}^D \mu_i^2 \nu_{1,i}^2 \right) \right]
        \\
        &= \exp\left(\frac{1}{2} \frac{n^2 c^4\epsilon^4}{D^2} \sum_{i = 1}^D \mu_i^2 \right) 
        \\
        &= \exp\left( \frac{n^2c^4 \epsilon^4}{2 D^2} \norm{M_D}_{F}^2 \right)
    \end{align*}
    using the fact that the Rademacher random variables $\nu_{1,i}, \nu_{2,i}$ are iid. sub-gaussian and applying Theorem 7.27 in \citet{Foucart2013}. 
    The main result of this proof is the fact that all the necessary information of the gentle measurement is encoded in the matrix $M_D$ together with the parameter vectors $\Phi_{\mathcal{V}_D}^{-1}(\Delta_\nu)$. For gentle measurements, the results from Proposition~\ref{prop::properties_of_Kraus_channel_2} give us a bound on the eigenvalues of $\bar{\mathcal{H}}$, which are the diagonal entries of $M_D$, in terms of the gentleness of the measurement. 
    \begin{proposition}
\label{prop::properties_of_Kraus_channel_2}
    Let $M_i = (M_{y_i})_{y_i \in \mathcal{Y}_i}$ be $\alpha$ gentle measurements for $i = 1,...,n$, $\alpha \in [0, 1/2)$ and $E_{y_i} = M_{y_i}^*M_{y_i}$. Furthermore, let $\bar{\mathcal{H}}$ be as in~\eqref{eqn::average_Kraus_channel} where the $\mathcal{H}_i$ are defined with respect to the gentle measurements $M_i$. Then
    \begin{enumerate}
        \item[(i)] There exists an orthonormal basis $V_1,...V_{d^2} \in \mathbb{H}_d$ of eigenvectors of $\bar{\mathcal{H}}$ with eigenvalues $\mu_1, ...,\mu_{d^2} \geq 0$ such that
        \begin{equation*}
            \bar{\mathcal{H}}(A) = \sum_{i = 1}^{d^2} \mu_i \Tr\left[V_i A \right] V_i
        \end{equation*}

        \item[(ii)] $\mathbbm{1}/\sqrt{d} = V_{d^2}$ is an eigenvector of $\bar{\mathcal{H}}$ with eigenvalue $1$.

        \item[(iii)] We have $\Tr\left[ V_i \right] = 0$ for $i = 1,...,d^2-1$.

        \item[(iv)] We have $\sum_{i = 1}^{d^2-1} \mu_i \leq \frac{16\alpha^2}{(1-4\alpha^2)^2}$.
    \end{enumerate}
\end{proposition}

\begin{proof}
    \begin{enumerate}
        \item[(i)] By Lemma \ref{lem::properties_of_Kraus_channel}, we know that $\bar{\mathcal{H}}$ is a positive $\mathbb{C}$-linear operator on the $\mathbb{C}$-vector space $\mathbb{C}^{d \times d}$. Since it is also Hermiticity-preserving, its restriction $\bar{\mathcal{H}}|_{\mathbb{H}_d}$ to the space of hermitian matrices $\mathbb{H}_d$ is also a self-adjoint and positive $\mathbb{R}$-linear operator on the $\mathbb{R}$-vector space $\mathbb{H}_d$. As such, there exist an orthonormal basis $V_1,...,V_{d^2} \in \mathbb{H}_d$ of eigenvectors of $\bar{\mathcal{H}}|_{\mathbb{H}_d}$ with eigenvalues $\Tilde{\mu}_1,...,\Tilde{\mu}_{d^2} > 0$ such that
        \begin{equation*}
            \bar{\mathcal{H}}|_{\mathbb{H}_d}(A) = \sum_{i = 1}^{d^2} \Tilde{\mu}_i \Tr[V_i A] V_i.
        \end{equation*}
        Since $\bar{\mathcal{H}}$ is a positive $\mathbb{C}$-linear operator on $\mathbb{C}^{d \times d}$, there exist an orthonormal basis of vectors $W_1,...,W_{d^2} \in \mathbb{C}^{d \times d}$ with eigenvalues $\mu_1,...,\mu_{d^2} > 0$ such that
        \begin{equation*}
            \bar{\mathcal{H}}(A) = \sum_{i = 1}^{d^2} \mu_i \Tr[W_i A] W_i.
        \end{equation*}
        Obviously, the vectors $V_i$ are eigenvectors of $\bar{\mathcal{H}}$ as well and they remain linear independent in the larger $\mathbb{C}$-vector space $\mathbb{C}^{d \times d}$ due to them remaining orthogonal. Therefore, it must hold $W_i = V_i$ and $\Tilde{\mu}_i = \mu_i$ for all $i$.

        \item[(ii)] Since $\bar{\mathcal{H}}$ is unital, $\mathbbm{1}$ is an eigenvector of $\bar{\mathcal{H}}$ with eigenvalue $1$ and therefore its renormalization $\mathbbm{1}/\sqrt{d} = V_{d^2}$ as well.

        \item[(iii)] Since the $V_i$ are orthogonal, we have
        \begin{equation*}
            0 = \Tr\left[ V_{d^2} V_i \right] = \Tr\left[ \mathbbm{1} V_i \right] = \Tr\left[ V_i \right]
        \end{equation*}
        for all $i = 1,...,d^2-1$.

        \item[(iv)] Since the measurements $M_i$ are $\alpha$-gentle, by Theorem~\ref{thm::duality_gentleness_privacy}, they are $\delta$-quantum-differentially-private for $\delta = 2\log(\frac{1+2\alpha}{1-2\alpha})$. Then each channel $\mathcal{H}_j$ has a basis of eigenvectors $V_i^{(j)}$ with eigenvalues $\mu_i^{(j)}$, where $V_{d^2}^{(j)} = \frac{1}{\sqrt{d}} \mathbbm{1}$. We know that the matrices $V_i^{(j)}$ are traceless and as such we have
        \begin{align*}
            \sum_{i = 1}^{d^2 -1} \mu_i^{(j)} &= \sum_{i = 1}^{d^2 -1} \Tr\left[ V_i^{(j)} \Bar{\mathcal{H}_j}(V_i^{(j)}) \right]
            \\
            &= \sum_{i = 1}^{d^2 -1} \sum_{y_j \in \mathcal{Y}_j} \Tr\left[ V_i^{(j)} E_{y_j} \right]^2 \frac{1}{\Tr[E_{y_j}]}
            \intertext{Now, using the fact that $\Tr[E_{y_j}] \geq d \lambda_{min}(E_{y_j})$ we get}
            &\leq \frac{1}{d} \sum_{i = 1}^{d^2 -1} \sum_{y_j \in \mathcal{Y}_j} \Tr\left[ V_i^{(j)} E_{y_j} \right]^2 \frac{1}{\lambda_{min}(E_{y_j})}
            \\
            &= \frac{1}{d} \sum_{i = 1}^{d^2 -1} \sum_{y_j \in \mathcal{Y}_j} \Tr\left[ V_i^{(j)} \left( \frac{E_{y_j}}{\lambda_{min}(E_{y_j})} - \mathbbm{1} \right) \right]^2 \lambda_{min}(E_{y_j})
            \\
            &= \frac{(e^{\delta} -1)^2}{d} \sum_{y_j \in \mathcal{Y}_j} \sum_{i = 1}^{d^2 -1} \Tr\left[ V_i^{(j)} A_{y_j} \right]^2 \lambda_{min}(E_{y_j})
            \intertext{where $A_{y_j} := \frac{1}{e^{\delta}-1} \left(\frac{E_{y_j}}{\lambda_{min}(E_{y_j})} - \mathbbm{1}\right)$ has maximal eigenvalue $1$ due to Proposition~\ref{prop::equalities_qDP}. Since $(V_i^{(j)})_{i = 1,...,d^2}$ forms an orthonormal basis of $\mathbb{C}^{d \times d}$, we may bound}
            &\sum_{i = 1}^{d^2 -1} \Tr\left[ V_i^{(j)} A_{y_j} \right]^2 \leq \norm{A_{y_j}}_{F}^2 \leq d.
            \intertext{Finally, this gives the bound}
            \sum_{i = 1}^{d^2 -1} \mu_i^{(j)} &\leq \frac{(e^{\delta} -1)^2}{d} \sum_{y_j \in \mathcal{Y}_j} d \lambda_{min}(E_{y_j}) \leq (e^{\delta} -1)^2.
        \end{align*}
        Therefore, using the fact that $\frac{1}{\sqrt{d}} \mathbbm{1}$ is an eigenvector for all $\mathcal{H}_j$ and $\Bar{\mathcal{H}}$ with eigenvalue 1, it holds 
        \begin{align*}
            \sum_{i = 1}^{d^2 -1} \mu_i &= \Tr\left[ \Bar{\mathcal{H}} \right] - 1 = \frac{1}{n} \sum_{j = 1}^n \Tr\left[ \mathcal{H}_j \right] - 1 = \frac{1}{n} \sum_{j = 1}^n \sum_{i = 1}^{d^2} \mu_i^{(j)} -1 
            \\
            &\leq \frac{1}{n} \sum_{j = 1}^n \left((e^{\delta} -1)^2 + 1\right) - 1 
            = (e^{\delta} -1)^2
            = \frac{16 \alpha^2}{(1-4\alpha^2)^2}.
        \end{align*}
    \end{enumerate}
    \end{proof}
    \noindent With the results of Proposition~\ref{prop::properties_of_Kraus_channel_2} we can now finalize the proof of Theorem~\ref{thm::lower_bound}. When choosing $D = \frac{d^2}{2}$ we can further bound
    \begin{align}
    \label{eqn::HS_norm_M_D_fixed}
        \norm{M_D}_{F}^2 = \sum_{i = 1}^D \mu_i^2 \leq D \mu_D^2 \leq D \left(\frac{\sum_{i = D}^{d^2 - 1} \mu_i}{d^2 -1 - (D-1)} \right)^2 
        \leq \frac{512 \alpha^4}{(1-2\alpha)^8 d^2}
    \end{align}
    where we used the fact that assume the eigenvalues $\mu_i$ to be ordered from smallest to largest and the result of Proposition~\ref{prop::properties_of_Kraus_channel_2} (iv). This gives
    \begin{align*}
        d_{\chi^2}\left( \mathbb{P}_{\rho_0^{\otimes n}}^{R^{M}}, \mathbb{E}_{\nu} \left[ \mathbb{P}_{\rho_{\nu}^{\otimes n}}^{R^{M}} \right] \right) &\leq \left( \frac{e^d}{e^d -2} \right)^2 \exp\left( 512 \frac{n^2 c^4 \epsilon^4 \alpha^4}{2D^2 d^2} \right) - 1 
        \\
        &\leq \left( \frac{e^d}{e^d -2} \right)^2 \exp\left( 1024 c^4 \frac{n^2 \epsilon^4 \alpha^4}{d^6} \right) - 1
    \end{align*}
    showing that the error of any locally-$\alpha$-gentle test is bounded from below from 0 as long as $n \geq \Omega\left( \frac{d^3}{\epsilon^2 \alpha^2} \right)$, completing the proof of Theorem~\ref{thm::lower_bound}. 
\end{proof}

For randomized measurements in the non-gentle case, \citet{Bubeck2020} has shown an optimal rate of $n = \Theta(\frac{d^{3/2}}{\epsilon^2})$. It turns out that our proof technique allows to show lower bounds for randomized measurements in the gentle case as in \citet{liu2024rolesharedrandomnessquantum}. If the measurements we perform are random, we cannot chose $\Delta_{\nu}$ in~\eqref{eqn::defn_alternative_states_2} according to the direction of least sensitivity of the measurement. This is because these directions are unknown to us due to the randomness of the measurement. Mathematically, this corresponds to the fact that we cannot assume the eigenvalues in~\eqref{eqn::HS_norm_M_D_fixed} to be ordered, resulting in a larger upper bound. In that case, for $D = d^2-1$, we have
\begin{equation}
\label{eqn::HS_norm_M_D_randomized}
    \sum_{i = 1}^{D} \mu_i^2 \leq \frac{16 \alpha^2}{(1-4\alpha^2)^2} \sum_{i = 1}^{D} \mu_i \leq \left(\frac{16 \alpha^2}{(1-4\alpha^2)^2}\right)^2 = \frac{256 \alpha^4}{(1-4\alpha)^4}
\end{equation}
which gives the following lower bound for randomized measurement schemes. We believe it to be an interesting open question to identify whether this lower bound is optimal.
\begin{corollary}
\label{thm::lower_bound_randomized}
    A total of $n = \Omega\left( \frac{d^2}{\epsilon^2 \alpha^2} \right)$ copies of the state $\rho$ are needed to verify whether $\rho$ is the maximally mixed state $\rho_0 = \frac{1}{d} \mathbbm{1}$ or $\norm{\rho - \rho_0}_{Tr} > \epsilon$ with high probability using randomized locally-$\alpha$-gentle measurements. 
\end{corollary}

Finally, we note that the proof is not exclusively valid for $\rho_0$ being the maximally mixed state. Suppose that $\rho_0$ is a full rank quantum state belonging to some class $\mathcal{S}_{c_{min}} = \left\{  \rho \in \mathcal{S}(\mathbb{C}^d) \; \middle| \;  \lambda_{min}(\rho) \geq \frac{c_{min}}{d}\right\}$ for some fixed $c_{min} > 0$. Then the results of Lemma~\ref{lem::probability_valid_states} still hold for $c = 10\sqrt{2}$ and $\epsilon < \frac{c_{min}}{c^2}$ showing that the construction in equation~\eqref{eqn::defn_alternative_states_2} are valid alternative states with high probability, even in the case that $\rho_0$ is not the maximally mixed state. Furthermore, using the properties of gentle measurements, for the probability mass functions defined in~\eqref{eqn::defn_p_and_q}, which we insert in~\eqref{eqn:relation_H_i_Lüders_channel}, we have
\begin{equation*}
    \frac{1}{p^{(i)}(y_i)} = \frac{1}{\Tr[\rho_0 E_{y_i}]} \leq \frac{1}{\lambda_{min}(E_{y_i})} \leq \frac{e^{\delta}}{\lambda_{max}(E_{y_i})} \leq \frac{d e^{\delta}}{\Tr[E_{y_i}]}.
\end{equation*}
Therefore, the linear super-operator we end up with in this case is the same as in~\eqref{eqn::average_Kraus_channel} with an additional factor of $e^{\delta}$. For $\alpha$ bounded away from $\frac{1}{2}$, we have $\delta$ bounded away from $\infty$ which does not alter the rate in terms of $d, n, \epsilon$ and $\alpha$.

\newpage
\bibliographystyle{apalike} 
\bibliography{testingquditsbiblio}       

\newpage
\appendix
\section{Proofs of Theorem~\ref{thm::duality_gentleness_privacy}}
\label{sec::appendix_gentleness_proofs}

This section is split into a first part in which we prove (i) in Theorem~\ref{thm::duality_gentleness_privacy} while the second part of this section is concerned with (ii) in Theorem~\ref{thm::duality_gentleness_privacy} and the improvement on the constants appearing for positive-definite measurement operators. Note that, since we work exclusively with locally-gentle measurements, it suffices to show the results on each register independently. For a brief overview on the differences between locally- and globally-gentle measurement, see \citet{butucea2025sampleoptimallearningquantumstates}. 

\begin{proposition}
\label{prop::equalities_qDP}
    Let $M = (M_{y})_{y \in \mathcal{Y}}$ be a quantum measurement. Then the following are equivalent:
    \begin{enumerate}
        \item[i)] $M$ is a $\delta$-quantum-differentially-private measurement on $\mathcal{S}(\mathbb{C}^d)$.

        \item[ii)] $M$ is a $\delta$-quantum-differentially-private measurement on $\mathcal{S}_{pure}(\mathbb{C}^d)$.

        \item[iii)] $\lambda_{max}(E_{y}) \leq e^{\delta} \lambda_{min}(E_{y})$ for all $y \in \mathcal{Y}$, where $E_{y} = M_{y}^*M_{y}$. 
    \end{enumerate}
\end{proposition} 

\begin{proof}
\noindent

\begin{enumerate}[left=0pt, leftmargin=20pt, align=left, itemindent=\parindent, itemsep=0pt, parsep=0pt, topsep=0pt]

\item[i) $\implies$ ii)] Since $\mathcal{S}_{pure}(\mathbb{C}^d) \subseteq \mathcal{S}(\mathbb{C}^d)$.

\item[ii) $\implies$ iii)] Let $M$ be $\delta$-quantum-differentially-private on pure states. For any $y \in \mathcal{Y}$ we define $E_{y} = M_{y}^*M_{y}$. Then for $\rho = \ket{\psi}\bra{\psi}, \rho' = \ket{\psi'}\bra{\psi'}$ it holds
    \begin{equation*}
        \bra{\psi}E_{y}\ket{\psi} = \mathbb{P}_{\ket{\psi}}(R^M = y) \leq e^{\delta}\mathbb{P}_{\ket{\psi'}}(R^M = y) = e^{\delta} \bra{\psi'}E_{y}\ket{\psi'}
    \end{equation*}
    Now, let $\ket{\psi}$ and $\ket{\psi'}$ be eigenvectors of $E_y$ associated to eigenvalue $\lambda_{max}(E_{y})$ and $\lambda_{min}(E_{y})$ respectively. Then we have
    \begin{align*}
        \lambda_{max}(E_{y}) = \bra{\psi}E_{y}\ket{\psi} &= \mathbb{P}_{\ket{\psi}}(R^M = y) 
        \\
        &\leq e^{\delta}\mathbb{P}_{\ket{\psi'}}(R^M = y) 
        \\
        &= e^{\delta} \bra{\psi'}E_{y}\ket{\psi'} = e^{\delta} \lambda_{min}(E_{y}).
    \end{align*}

\item[iii) $\implies$ i)] Note that we can write any quantum state as $\rho = \sum_{n = 1}^{d} \lambda_n \ket{\psi_n}\bra{\psi_n}$. The outcome probability for $y \in \mathcal{Y}$ is then given by
    \begin{equation*}
        \mathbb{P}_{\rho}\left( R^M = y \right) = \Tr\left[ \rho E_{y} \right] = \sum_{n = 1}^{d} \lambda_k \Tr\left[ \ket{\psi_k}\bra{\psi_k} E_{y} \right] = \sum_{n = 1}^{d} \lambda_k \bra{\psi_k}E_{y}\ket{\psi_k}.
    \end{equation*}
    By definition of $\lambda_{max}(E_{y})$ and $\lambda_{min}(E_{y})$, since $\sum_{k = 1}^{d} \lambda_k = 1$, we have
    \begin{equation*}
        \lambda_{min}(E_{y}) \leq \mathbb{P}_{\rho}\left( R^M = y \right) \leq \lambda_{max}(E_{y})
    \end{equation*}
    for any $\rho \in \mathcal{S}(\mathbb{C}^d)$ . Therefore we have
    \begin{equation*}
        \mathbb{P}_{\rho}\left( R^M = y \right) \leq \lambda_{max}(E_{y}) \leq e^{\delta} \lambda_{min}(E_{y}) \leq e^{\delta} \mathbb{P}_{\rho'}\left( R^M = y \right)
    \end{equation*}
    for all $\rho, \rho' \in \mathcal{S}(\mathbb{C}^d), y \in \mathcal{Y}$.
\end{enumerate}
\end{proof}

\begin{proposition}
\label{prop::Privacy_implies_gentleness_pure}
    Let $M = (M_y)_{y \in \mathcal{Y}}$ be a quantum measurement and $E_y = M_y^*M_y$. If $M$ is $\delta$-quantum-differentially-private on $\mathcal{S}(\mathbb{C}^{d})$, then there exists an implementation $\Tilde{M}$ of $M$ such that $\Tilde{M}$ is $\alpha$-gentle on $\mathcal{S}_{pure}(\mathbb{C}^{d})$ for
    \begin{equation*}
        \alpha = \frac{e^{\frac{\delta}{2}} -1}{e^{\frac{\delta}{2}} +1} = \tanh\left( \frac{\delta}{4} \right).
    \end{equation*}
\end{proposition}
\begin{proof}
    Note that, by Proposition~\ref{prop::equalities_qDP}, we have $\lambda_{max}(E_{y}) \leq e^{\delta} \lambda_{min}(E_{y})$ for all $y \in \mathcal{Y}$, where $E_{y} = M_{y}^*M_{y}$. Let $|M_y| = \sqrt{E_y}$ be the unique positive square root of $E_y$. Then $\Tilde{M} = (|M_y|)_{y |\in \mathcal{Y}}$ does define a quantum measurement that has the same outcome probabilities as $M_y$. Given $|M_y|$, there exists an orthonormal basis $\ket{v_{y,1}},...,\ket{v_{y,d}}$ of $\mathbb{C}^d$ and $\lambda_{y,1},...,\lambda_{y,d} > 0$ such that
    \begin{equation*}
        |M_y| = \sum_{i = 1}^d \lambda_{y, i} \ket{v_{y, i}}\bra{v_{y, i}},
    \end{equation*}
    where $\lambda_{y,1} = \sqrt{\lambda_{min}(E_y)}$ and $\lambda_{y,d} = \sqrt{\lambda_{max}(E_y)}$. Any pure state $\ket{\psi} \in \mathcal{S}_{pure}(\mathbb{C}^d)$ can be written as
    \begin{equation*}
        \ket{\psi} = \sum_{i = 1}^d \mu_i \ket{v_{y,i}}, \hspace{20pt} \text{for } \sum_{i = 1}^d |\mu_i|^2 = 1.
    \end{equation*}
    Then, for $\rho = \ket{\psi}\bra{\psi}$ we have
    \begin{align*}
        \norm{\rho - \rho_{|M| \to y}}_{Tr}^2 = 1 - \frac{\left|\bra{\psi}|M_y|\ket{\psi}\right|^2}{\bra{\psi}|M_y|^2\ket{\psi}}\leq 1 - \frac{4 \lambda_1 \lambda_d}{(\lambda_d+ \lambda_1)^2} = \left(\frac{\lambda_d - \lambda_1}{\lambda_d + \lambda_1}\right)^2 \leq \left(\frac{e^{\frac{\delta}{2}} -1}{e^{\frac{\delta}{2}} +1}\right)^2,
    \end{align*}
    where the first inequality is due to the Kantorovich inequality (\citet{Moslehian2012}) and the second inequality is due to the fact that the eigenvalues of $E_y$ are given by the square of the eigenvalues of $|M_y|$. This shows that $\Tilde{M}$ is $\alpha$-gentle on pure states.
\end{proof}

\begin{proposition}
\label{prop::pure_states_imply_all_states}
    Let $M = (M_{y})_{y \in \mathcal{Y}}$ be an $\alpha$-gentle measurement on $\mathcal{S}_{pure}(\mathbb{C}^d)$ such that $M_{y}$ is positive and self-adjoint. Then, $M$ is $\alpha$-gentle on $\mathcal{S}(\mathbb{C}^d)$.
\end{proposition}

\begin{proof}
    Let $\rho = \sum_{k = 1}^{d}\lambda_k \ket{\psi_k}\bra{\psi_k} \in \mathcal{S}(\mathbb{C}^d)$ be any quantum state and $M$ a $\alpha$-gentle measurement on $\mathcal{S}_{pure}(\mathbb{C}^d)$. Define \begin{equation*}
        \ket{\Psi} = \sum_{k = 1}^{d} \sqrt{\lambda_k} \ket{\psi_k} \otimes \ket{\psi_k} \in \mathcal{S}_{pure}(\mathbb{C}^d \otimes \mathbb{C}^d).
    \end{equation*} 
    Then it holds $\rho = \Tr_2[\ket{\Psi}\bra{\Psi}]$, where $\Tr_2$ is the partial trace over the seconds Hilbert space $\mathbb{C}^d$. Furthermore, for the measurement $M \otimes I = (M_{y} \otimes I)_{y \in \mathcal{Y}}$ it holds
    \begin{align*}
        \mathbb{P}_{\rho}\left( R^M = y \right) = \Tr\left[ \rho M_{y}^* M_{y} \right] &= \Tr\left[ \Tr_2[\ket{\Psi}\bra{\Psi}] M_{y}^* M_{y} \right] 
        \\
        &= \Tr\left[ \ket{\Psi}\bra{\Psi} (M_{y}^*M_{y} \otimes I) \right]
        \\
        &= \Tr\left[ \ket{\Psi}\bra{\Psi} (M_{y}^* \otimes I) (M_{y} \otimes I) \right]
        \\
        &= \Tr\left[ \ket{\Psi}\bra{\Psi} (M_{y} \otimes I)^* (M_{y} \otimes I) \right] = \mathbb{P}_{\ket{\Psi}}\left( R^{M\otimes I} = y \right).
    \end{align*}
    Furthermore, we have 
    \begin{align*}
        M_{y} \rho M_{y}^* = M_{y} \Tr_2[\ket{\Psi}\bra{\Psi}] M_{y}^* = \Tr_2 \left[ (M_{y} \otimes I) \ket{\Psi}\bra{\Psi} (M_{y}^* \otimes I) \right].
    \end{align*}
    This shows that
    \begin{align*}
        \rho_{M \to y} &= \frac{1}{\mathbb{P}_{\rho}\left( R^M = y \right)} M_{y}\rho M_{y}^* 
        \\
        &= \frac{1}{\mathbb{P}_{\ket{\Psi}}\left( R^{M\otimes I} = y \right)} \Tr_2 \left[ (M_{y} \otimes I) \ket{\Psi}\bra{\Psi} (M_{y}^* \otimes I) \right]
        \\
        &= \Tr_2 \left[ \frac{1}{\sqrt{\mathbb{P}_{\ket{\Psi}}\left( R^{M\otimes I} = y \right)}} (M_{y} \otimes I) \ket{\Psi}\bra{\Psi} (M_{y}^* \otimes I) \frac{1}{\sqrt{\mathbb{P}_{\ket{\Psi}}\left( R^{M\otimes I} = y \right)}} \right]
        \\
        &= \Tr_2\left[ \ket{\Psi_{M \otimes I \to y}}\bra{\Psi_{M \otimes I \to y}} \right].
    \end{align*}
    Since the trace norm is contractive under quantum channels such as the partial trace, we have
    \begin{align*}
        \norm{\rho - \rho_{M \to y}}_{Tr} &= \norm{\Tr_2[\ket{\Psi}\bra{\Psi}] - \Tr_2\left[ \ket{\Psi_{M \otimes I \to y}}\bra{\Psi_{M \otimes I \to y}} \right]}_{Tr} 
        \\
        &\leq \norm{\ket{\Psi}\bra{\Psi} -  \ket{\Psi_{M \otimes I \to y}}\bra{\Psi_{M \otimes I \to y}}}_{Tr}.
    \end{align*}
    Therefore, the gentleness of $M$ on $\mathcal{S}(\mathbb{C}^d)$ is bounded by the gentleness of $M \otimes I$ on $\mathcal{S}_{pure}(\mathbb{C}^d \otimes \mathbb{C}^d)$. By Lemma~\ref{lem::improved_constant_positivity}, $M$ is $\delta$-quantum-differentially-private on $\mathcal{S}(\mathbb{C}^d)$ for $\delta = 4 \arctanh(\alpha)$. Note that even though we have not proven Lemma 14 yet, its proof is independent of the results shown so far. It therefore remains applicable. We now have
    \begin{equation*}
        \lambda_{max}(M_{y}^*M_{y} \otimes I) = \lambda_{max}(M_{y}^*M_{y}) \leq e^{\delta} \lambda_{min}(M_{y}^*M_{y}) = e^{\delta} \lambda_{min}(M_{y}^*M_{y} \otimes I)
    \end{equation*}
    which shows that $M \otimes I$ is also $\delta$-quantum-differentially-private on $\mathcal{S}(\mathbb{C}^d \otimes \mathbb{C}^d)$. Proposition \ref{prop::Privacy_implies_gentleness_pure} then shows, that $M \otimes I$ is $\alpha$-gentle on $\mathcal{S}_{pure}(\mathbb{C}^d \otimes \mathbb{C}^d)$ from which we get
    \begin{align*}
        \norm{\rho - \rho_{M \to y}}_{Tr} \leq \norm{\ket{\Psi}\bra{\Psi} -  \ket{\Psi_{M \otimes I \to y}}\bra{\Psi_{M \otimes I \to y}}}_{Tr} \leq \alpha.
    \end{align*}
\end{proof}

\begin{corollary}
\label{prop::Privacy_implies_gentleness}
    Let $M = (M_y)_{y \in \mathcal{Y}}$ be a quantum measurement and $E_y = M_y^*M_y$. If $M$ is $\delta$-quantum-differentially-private on $\mathcal{S}(\mathbb{C}^{d})$, then there exists an implementation $\Tilde{M}$ of $M$ such that $\Tilde{M}$ is $\alpha$-gentle on $\mathcal{S}(\mathbb{C}^{d})$ for
    \begin{equation*}
        \alpha = \frac{e^{\frac{\delta}{2}} -1}{e^{\frac{\delta}{2}} +1} = \tanh\left( \frac{\delta}{4} \right).
    \end{equation*}
\end{corollary}

\begin{proof}
    For the implementation $\Tilde{M} = (|M_{y}|)_{y \in \mathcal{Y}}$ we chose in the proof of Proposition~\ref{prop::Privacy_implies_gentleness_pure}, we have that the measurement operators are positive and self-adjoint. Thus, by Proposition~\ref{prop::pure_states_imply_all_states}, we have that the same implementation $\Tilde{M}$ is $\alpha$-gentle on $\mathcal{S}(\mathbb{C}^d)$.
\end{proof}

This concludes the proof of part (i) of Theorem~\ref{thm::duality_gentleness_privacy} which shows that quantum-differentially-private measurements have a gentle implementation. The following two results are now concerned with the opposite direction showing that gentle measurements are always quantum-differentially-private.

\begin{proposition}
\label{lem::gentleness_implies_privacy}
    Let $\alpha $ in $[0, \frac{1}{2})$ and $M$ be $\alpha$-gentle on $\mathcal{S}(\mathbb{C}^d)$ with measurement operators $M_{y}$. Then $M$ is $\delta$ quantum differentially-private on $\mathcal{S}(\mathbb{C}^d)$ for $\delta = 2\log(\frac{1+2\alpha}{1-2\alpha})$. 
\end{proposition}

\begin{proof}
    Assume that $M = (M_{y})_{y \in \mathcal{Y}}$ is an $\alpha$-gentle measurement on $\mathcal{S}(\mathbb{C}^d)$. Define $E_{y} := M_{y}^*M_{y}$. Let $\rho_1, \rho_2 \in \mathcal{S}(\mathbb{C}^d)$ s.t. $\norm{\rho_1 - \rho_2}_{Tr} = 1$. Let us further denote by $p_1$ (respectively $p_2$) the probability of obtaining outcome $y$ under $\rho_1$ (respectively $\rho_2$). That is
    \begin{equation*}
        p_1 = \mathbb{P}_{\rho_1}\left(R^M = y \right) \hspace{20pt} \text{and} \hspace{20pt} p_2 = \mathbb{P}_{\rho_2}\left(R^M = y \right) 
    \end{equation*} Without loss of generality we assume that $p_1 > p_2 \geq 0$. Now, let
    \begin{equation*}
        \rho_{\lambda} = \lambda \rho_1 + (1-\lambda) \rho_2 \hspace{10pt} \text{for all } \lambda \in (0,1).
    \end{equation*}
    The probability of obtaining the outcome $y$ when measuring $\rho_{\lambda}$ is
    \begin{equation*}
        p_{\lambda} = \mathcal{P}_{\rho_{\lambda}}(R^M = y) = \Tr\left( \rho_{\lambda} E_y \right) = \lambda \Tr\left( \rho_1 E_y \right) + (1 - \lambda) \Tr\left( \rho_2 E_y \right) = \lambda p_1 + (1 - \lambda) p_2.
    \end{equation*}
    The post-measurement state of $\rho_{\lambda}$ is then given by
    \begin{equation*}
        (\rho_{\lambda})_{M \to y} = \frac{1}{p_{\lambda}} M_y \rho_{\lambda} M_y^* = \frac{\lambda p_1 (\rho_1)_{M \to y} + (1-\lambda) p_2 (\rho_2)_{M \to y}}{\lambda p_1 + (1-\lambda) p_2}.
    \end{equation*}
    Now if we define $\delta = \frac{\lambda p_1}{\lambda p_1 + (1-\lambda) p_2} - \lambda > 0$, we get
    \begin{align*}
        \rho_{\lambda} - (\rho_{\lambda})_{M \to y} = \frac{\lambda p_1}{p_{\lambda}} \left((\rho_1 - (\rho_1)_{M \to y}\right) + \frac{(1-\lambda) p_2}{p_{\lambda}} \left(\rho_2 - (\rho_2)_{M \to y}\right) + \delta \left(\rho_2 - \rho_1\right)).
    \end{align*}
    By the triangle inequality and gentleness we now have
    \begin{align*}
        \delta \norm{\rho_2 - \rho_1}_{Tr} \leq \frac{\lambda p_1}{p_{\lambda}} \alpha + \frac{(1-\lambda) p_2}{p_{\lambda}} \alpha + \alpha = 2 \alpha
    \end{align*}
    Since we further assumed $\norm{\rho_2 - \rho_1}_{Tr} = 1$, we get $\delta \leq 2 \alpha$. This allows us to write
    \begin{align*}
        p_1 = \frac{\lambda - \lambda^2 + \delta (1-\lambda)}{\lambda - \lambda^2 - \delta \lambda} p_2 \leq \frac{\lambda - \lambda^2 + 2\alpha (1-\lambda)}{\lambda - \lambda^2 - 2 \alpha \lambda} p_2 \hspace{10pt} \text{for all } \lambda \in (0,1 - 2 \alpha).
    \end{align*}
    The last inequality only holds as long as the denominator is positive which is the case for $\lambda < 1 - 2 \alpha$. As such, for $\lambda_0 = \frac{1 - 2 \alpha}{2} < 1 - 2 \alpha$, we obtain
    \begin{equation*}
        p_1 \leq \frac{\lambda_0 - \lambda_0^2 + 2\alpha (1-\lambda_0)}{\lambda_0 - \lambda_0^2 - 2 \alpha \lambda_0} p_2 = \left( \frac{1 + 2\alpha}{1 - 2\alpha} \right)^2 p_2.
    \end{equation*}
    Now, since we started with $\rho_1, \rho_2$ such that $\norm{\rho_1 - \rho_2}_{Tr} = 1$, the last relation holds for every pure state. As such $M$ is $2 \log\left( \frac{1 + 2\alpha}{1 - 2\alpha} \right)$-quantum-differentially-private on pure states. Proposition \ref{prop::equalities_qDP} then proves that $M$ is $2 \log\left( \frac{1 + 2\alpha}{1 - 2\alpha} \right)$-quantum-differentially-private on the whole space.
\end{proof}

While Proposition~\ref{lem::gentleness_implies_privacy} shows that an arbitrary gentle measurement is quantum differentially private, showing part (ii) of Theorem~\ref{thm::duality_gentleness_privacy}, we can show that the constant relating the two can be improved for measurements with positive-definite measurement operators.

\begin{lemma}
\label{lem::improved_constant_positivity}
    Let $\alpha $ in $[0, 1)$ and $M$ be $\alpha$-gentle on $\mathcal{S}(\mathbb{C}^d)$ with positive-definite measurement operators $M_y$. Then $M$ is $\delta$ quantum differentially-private on $\mathcal{S}(\mathbb{C}^d)$ for $\delta = 2\log(\frac{1+\alpha}{1-\alpha}) = 4 \arctanh(\alpha)$.
\end{lemma}
    
\begin{proof}
    Let $M_y = \sum_{i = 1}^d \lambda_i \ket{v_i}\bra{v_i}$ be positive-definite with maximal and minimal eigenvalue $\lambda_d$ and $\lambda_1$ respectively. Consider the gentleness of $M_y$ on the pure state $\ket{\psi} = \frac{1}{\sqrt{\lambda_1 + \lambda_d}} (\sqrt{\lambda_d} \ket{v_1} + \sqrt{\lambda_1} \ket{v_d})$. Lemma \ref{lemmaTraceNorm} then shows that, due to the gentleness of $M$ on $\rho = \ket{\psi}\bra{\psi}$, we have
    \begin{align*}
        \alpha^2 \geq \norm{\rho - \rho_{M \to y}}_{Tr}^2 = 1- \frac{\left| \bra{\psi} M_y \ket{\psi} \right|^2}{\bra{\psi} M_y^2 \ket{\psi}}.
    \end{align*}
    For the numerator and denominator we have
    \begin{align*}
        \left|\bra{\psi} M_y \ket{\psi}\right|^2 &= \frac{1}{(\lambda_1 + \lambda_d)^2} \left( \lambda_d \lambda_1 + \lambda_d \lambda_1 \right)^2 = \frac{4 \lambda_1^2 \lambda_d^2}{(\lambda_1+ \lambda_d)^2}
        \\
        \bra{\psi} M_y^2 \ket{\psi} &= \frac{1}{\lambda_1 + \lambda_d} \left( \lambda_d \lambda_1^2 + \lambda_d^2 \lambda_1 \right) = \lambda_1 \lambda_d.
    \end{align*}
    from which we obtain
    \begin{equation*}
        \frac{\left| \bra{\psi} M_y \ket{\psi} \right|^2}{\bra{\psi} M_y^2 \ket{\psi}} = \frac{4 \lambda_1 \lambda_d}{(\lambda_1 + \lambda_d)^2}
    \end{equation*}
    This gives
    \begin{equation*}
        \alpha^2 \geq \frac{(\lambda_d - \lambda_1)^2}{(\lambda_d + \lambda_1)^2} \hspace{20pt} \text{or equivalently} \hspace{20pt} \alpha \geq \frac{\lambda_d - \lambda_1}{\lambda_d + \lambda_1} = \frac{\frac{\lambda_d}{\lambda_1} - 1}{\frac{\lambda_d}{\lambda_1} +1 }
    \end{equation*}
    Using the fact that for $E_{y} = M_y^2$ we have $\lambda_{max}(E_{y}) = \lambda_d^2$ and $\lambda_{min}(E_{y}) = \lambda_1^2$, we get
    \begin{align*}
        \alpha \geq \frac{e^{2 \frac{1}{4}\log\left( \frac{\lambda_{max}(E_{y})}{\lambda_{min}(E_{y})} \right)} - 1}{e^{2 \frac{1}{4}\log\left( \frac{\lambda_{max}(E_{y})}{\lambda_{min}(E_{y})} \right)} + 1} = \tanh\left( \frac{1}{4} \log\left( \frac{\lambda_{max}(E_{y})}{\lambda_{min}(E_{y})} \right) \right).
    \end{align*}
    Finally, using the monotonicity of $\tanh$, we get
    \begin{equation*}
        \frac{\lambda_{max}(E_{y})}{\lambda_{min}(E_{y})} \leq e^{4 \arctanh(\alpha)}
    \end{equation*}
    which together with Proposition \ref{prop::equalities_qDP} shows that $M$ is $\delta$-quantum-differentially-private on $\mathcal{S}(\mathbb{C}^d)$ for
    \begin{equation*}
        \delta = 4 \arctanh(\alpha).
    \end{equation*}
\end{proof}

\end{document}